\documentclass[aps,prl,reprint,superscriptaddress,10pt,twocolumn]{revtex4-2}
\usepackage[ascii]{inputenc}
\usepackage{fontenc}
\usepackage{amsmath}
\usepackage{amsfonts,amssymb,amsthm, bbm,braket,mathtools}
\usepackage{hyperref}
\usepackage[capitalize]{cleveref}
\usepackage{mathtools}
\usepackage{mleftright}
\usepackage[dvipsnames]{xcolor}
\usepackage{graphicx}
\usepackage{subfigure}
\usepackage{amsbsy}
\usepackage[bold]{hhtensor}
\usepackage{natbib}
\usepackage{times}
\usepackage{multirow}
\newtheorem{theorem}{Theorem}
\newtheorem{corollary}[theorem]{Corollary}
\newtheorem{lemma}[theorem]{Lemma}

\newcommand{\dens}[1]{\ket{#1}\!\!\bra{#1}}
\newcommand{\braakett}[2]{\langle\!\langle #1|#2\rangle\!\rangle}
\newcommand{\kett}[1]{|{#1}{\rangle\!\rangle}}
\newcommand{\braa}[1]{{\langle\!\langle}{#1}|}

\newcommand{\proj}[1]{\ket{#1} \!\! \bra{#1}}

\DeclareMathOperator{\Tr}{Tr}

\DeclareMathOperator\tr{Tr}

\newcommand{\CC}{\mathbb{C}}

\newcommand{\EE}{\mathbb{E}}
\newcommand{\VV}{\mathbb{V}}
\newcommand{\PP}{\mathbb{P}}
\newcommand{\mc}[1]{\mathcal{#1}}

\newcommand{\ct}{{}^\dagger}
\renewcommand{\Pr}{\operatorname{\PP}}

\newcommand{\tn}[1]{^{\otimes#1}}
\newcommand{\norm}[1]{\left\Vert #1 \right\Vert}
\DeclarePairedDelimiter\abs{\lvert}{\rvert}

\DeclarePairedDelimiter\ceil{\lceil}{\rceil}
\DeclarePairedDelimiter\parens{\lparen}{\rparen}
\DeclarePairedDelimiter\braces{\lbrace}{\rbrace}
\DeclarePairedDelimiter\bracks{\lbrack}{\rbrack}

\newcommand{\ot}{\otimes}
\newcommand{\eps}{\varepsilon}
\renewcommand{\epsilon}{\varepsilon}
\newcommand{\HS}{\text{HS}}
\newcommand{\T}{\mathrm{T}}
\allowdisplaybreaks[4]

\begin{document}

\title{Thrifty shadow estimation: re-using quantum circuits and bounding tails}
\author{Jonas Helsen}
\affiliation{QuSoft and CWI, The Netherlands}
\author{Michael Walter}
\affiliation{Faculty of Computer Science, Ruhr University Bochum, Germany}
\begin{abstract}
Shadow estimation is a recent protocol that allows estimating exponentially many expectation values of a quantum state from ``classical shadows'', obtained by applying random quantum circuits and computational basis measurements.
In this paper we study the statistical efficiency of this approach in light of near-term quantum computing.
We propose a more practical variant of the protocol, \emph{thrifty shadow estimation}, in which quantum circuits are reused many times instead of having to be freshly generated for each measurement.
We show that reuse is maximally effective when sampling Haar random unitaries, and maximally ineffective when sampling from the Clifford group, i.e., one should not reuse circuits when performing shadow estimation with the Clifford group. We provide an efficiently simulable family of quantum circuits that interpolates between these extremes, which we believe should be used instead of the Clifford group.
Finally, we consider tail bounds for shadow estimation and discuss when median-of-means estimation can be replaced with standard mean estimation.
\end{abstract}
\maketitle

A key aspect of the development of larger-scale quantum computers is the availability of protocols that can diagnose errors and noise in quantum computations.
Over the years many such protocols  have been proposed, optimizing either for informational completeness (various forms of tomography) or sampling efficiency (e.g., randomized benchmarking~\cite{magesan2011scalable,helsen2020general} or direct fidelity estimation~\cite{helsen2020general}), but not achieving both at the same time. See~\cite{eisert2020quantum} for a general overview.
Recently Huang, K\"ung and Preskill (HKP) went beyond this apparent dichotomy by proposing \emph{shadow estimation}~\cite{huang2020predicting}, a randomized protocol which extracts \emph{exponentially} many expectation values $\tr(O\rho)$ from \emph{polynomially} many copies of the state~$\rho$, with the only caveat being a restriction on the set of allowed observables~$O$~\footnote{The name being derived from a more theoretical proposal due to Aaronson~\cite{aaronson2018shadow}.}.
Shadow estimation has generated significant interest and led to several theoretical follow-up works~\cite{zhao2020fermionic,hu2021classical} and experimental applications~\cite{huggins2022unbiasing,struchalin2021experimental} (see also~\cite{elben2022randomized} for a comprehensive overview).
At its face value, the protocol is extremely simple:
upon receiving a state~$\rho$ one generates a random $n$-qubit circuit~$U$ from a circuit set~$\mathbb{U}$, applies it to the state~$\rho$, and then measures in the computational basis, obtaining a bit string~$x$.
The tuple~$(U,x)$ then forms a so-called \emph{classical shadow} of the state~$\rho$, from which expectation values~$\tr(O\rho)$ can be reconstructed by classical post-processing (see \cref{fig:shadow}~(a) for details).
The performance of the protocol depends on the circuit set $\mathbb{U}$, as well as the observables one considers.
An important case is when the circuit set is the multi-qubit Clifford group~$\mathbb{C}_n$, which is a $3$-design.
In this case, shadow tomography is efficient for observables~$O$ for which~$\tr(O^2)$ is bounded.
The Clifford group furthermore has the advantage that if the observable is, e.g., a projection onto a stabilizer state, then the classical post-processing needed is also \emph{computationally efficient} by the Gottesman-Knill theorem \cite{aaronson2004improved}, which is very useful in practice.

A key component of the HKP proposal is that \emph{every classical shadow requires an independent random circuit}.
This is critical to the mathematical argument for its statistical efficiency, but can be undesirable in practice.
Especially in near-term quantum computers, it is preferable in many systems to measure a fixed circuit multiple times to generate a large number of classical shadows.
This can already be seen in experimental implementations of shadow estimation such as~\cite{struchalin2021experimental}, which reports repeating each circuit $>10^{4}$ times, and~\cite{huggins2022unbiasing}, which reports measuring each circuit~$10^3$ times.
This is likely inspired by experience with randomized benchmarking, which similarly samples random circuits (and where the statistics of repeating circuits is well understood~\cite{wallman2014randomized,helsen2019multiqubit}).
In this work we systematically study the effect of circuit repetition for shadow tomography, which to the best of our knowledge has not been studied before, using tools from representation theory. We also apply those tools to the question of whether median-of-means estimators, another key component of the HKP proposal, are actually necessary for shadow tomography.

\begin{figure*}
  \noindent\fbox{\parbox{0.47\linewidth}{\raggedright
  \textbf{(a) Original shadow estimation protocol~\cite{huang2020predicting}:}

  \smallskip

  \emph{Data acquisition:} For $t=1,\dots,N$:
  \begin{enumerate}
  \item Draw a random circuit $\hat{U}_t \in \mathbb{U}$.
  \item Prepare $\rho$, apply the unitary $\hat{U}_t$, and measure in the computational basis.
    Record the outcome $\hat{x}_t \in \{0,1\}^n$.
  \end{enumerate}

  \emph{Prediction of expectation values:}
  \begin{enumerate}
  \item For $k=1,\dots,K$, compute the batch mean
  \[ \hat\rho_{(k)} = \frac{K}{N} \! \sum_{t=(k-1){\frac NK}+1}^{k{\frac NK}} \!\!\!\! \mathcal{F}^{-1} \bigl( \hat{U}_t^\dagger \proj{\hat{x}_t} \hat{U}_t \bigr). \]
  \item For $O \in \mathbb{O}$, output $\hat o = \operatorname{median} \, \{ \tr(O \hat\rho_{(k)}) \}_{k=1}^K$.
  \end{enumerate}
  }}
  \noindent\fbox{\parbox{0.47\linewidth}{\raggedright
  \textbf{(b) Thrifty shadow estimation (this work):}

  \smallskip

  \emph{Data acquisition:} For $t=1,\dots, N/R$:
  \begin{enumerate}
  \item Draw a random circuit $\hat{U}_t \in \mathbb{U}$.
  \item For $r=1\dots,R$:
  Prepare $\rho$, apply $\hat{U}_t$, and measure in the computational basis. Record the outcome $\hat{x}_{t,r} \in \{0,1\}^n$.
  \end{enumerate}

  \emph{Prediction of expectation values:}
  \begin{enumerate}
  \item For $k=1,\dots,K$, compute the batch mean
  \[ \!\!\!\!\!\!\!\!\hat\rho_{(k)} = \frac{RK}{N} \!\sum_{t=(k\!-\!1)\!{\frac N {RK}}\!+\!1}^{k{\frac N {RK}}} \!\!\frac1R \sum_{r=1}^R \mathcal{F}^{-1} \bigl( \hat{U}_t^\dagger \proj{\hat{x}_{t,r}} \hat{U}_t \bigr). \]
  \item For $O \in \mathbb{O}$, output $\hat o_R = \operatorname{median} \, \{ \tr(O \hat\rho_{(k)}) \}_{k=1}^K$.
  \end{enumerate}
  }}
  \caption{\label{fig:shadow}%
  (a)~\emph{The shadow estimation protocol of~\cite{huang2020predicting}:}
  A total of~$N$ measurements is performed, and each random circuit is used to obtain a single quantum measurement outcome.
  The parameter~$K$ corresponds to the number of batches in the median-of-means estimator.
  We assume that~$N$ is a multiple of~$K$.
  The quantum channel~$\mc{F}$ depends on the circuit set, and is given explicitly in \cref{eq:frame_op}.
  (b)~\emph{Thrifty shadow estimation as introduced in this work:}
  Each random circuit is re-used~$R$ times.
  We again use~$N$ to indicate the total number of measurements, and thus~$N/R$ random circuits are generated.
  The parameter~$K$ again corresponds to the number of batches of the median-of-means estimator.
  We assume here that~$N$ is a multiple of~$RK$.
  Note that we sample at least one random circuit per batch, as required for the median-of-means estimator.}
\end{figure*}

\section{Thrifty shadow estimation}
We introduce \emph{thrifty shadow estimation}, our variant of shadow estimation that re-uses quantum circuits and can be significantly more economic in practice.
The standard and thrifty protocols protocols are summarized in \cref{fig:shadow}.

We write~$N$ for the total number of measurements, $R$ for the number of times that a random circuit is re-used (including the first time) and~$K$ for the number of batches in the median-of-means estimator.
We will assume that $N$ is a multiple of~$KR$ throughout this section.
Thus, the protocol uses $N/R$ random quantum circuits.
Note that thrifty shadow estimation reduces to ordinary shadow estimation for~$R=1$.
To analyze its statistical performance, note that the thrifty estimator $\hat{o}_R$ is the median-of-means estimator for $N/R$ many i.i.d.\ copies of the random variable
\begin{equation}\label{eq:referenced in supplemental}
\mathbf{X}_R = \frac1R \sum_{r=1}^R \mathbf{X}^{(r)},
\end{equation}
where~$\mathbf{X}^{(r)} = \tr\big(O \mathcal F^{-1} (\mathbf{U}\dens{\mathbf{x}^{(r)}} \mathbf{U}\ct)\big)$,
with~$\mathbf{U}$ drawn uniformly at random from the circuit set~$\mathbb{U}$, and~$\mathbf{x}^{(1)},\dots,\mathbf{x}^{(R)}$ drawn i.i.d.\ from the conditional distribution
\begin{align*}
  p(x|U) = \bra{x} U \rho U\ct \ket{x}.
\end{align*}
Finally, $\mathcal{F}^{-1}$ is the inverse of the quantum channel~$\mathcal F$,
\begin{align}\label{eq:frame_op}
  \mathcal{F}(A) := \sum_{x\in \{0,1\}^n} \mathbb{E}_{U\in \mathbb{U}}  U\ct \dens{x}U \bra{a}UAU\ct \ket{x} ,
\end{align}
associated with the circuit set $\mathbb{U}$ and the computational basis POVM~%
\footnote{The quantum channel~$\mathcal F$ is also called a \emph{frame operator}. While it need not be invertible in general, it is invertible for all circuit sets considered in this Letter (as well as those considered in all other shadow tomography protocols that we are aware of).}.
It can be shown that $\EE (\mathbf{X}_R) = \Tr(O \rho)$.

This directly suggests the following protocol for estimating expectation values~$\Tr(O \rho)$: sample $R$ times from the distribution $p(x|U)$, compute the corresponding states~$\mathcal{F}^{-1} (U\ct \dens{x} U)$ and construct an estimator for the mean.
Concretely, it was shown (for standard shadow tomography) in \cite{huang2020predicting} that if one obtains~$N$ random samples~$\{(\hat{U}_t, \hat{x}_t)\}_{t=1}^N$, corresponding to $N$~independent random circuits, groups those into~$K$ equal-size batches, and computes the \emph{median-of-means} estimator~$\hat o$ as in \cref{fig:shadow}~(a),
then one can obtain with high probability an accurate estimate of the desired expectation value. This directly generalizes to thrifty shadow estimation.

In particular, if we set the batch size~$K = \ceil{8\log(1/\delta)}$ for~$\delta \in (0,1)$, then \cite[Theorem~2]{lugosi2019mean} implies that for any fixed observable~$O$,
\begin{align*}
  \abs[\big]{\hat o_R - \Tr(O \rho)} \leq \sqrt{\frac{32 \, \mathbb{V}_R(O,\rho)\log(1/\delta)}{N/R}},
\end{align*}
with probability at least~$1-\delta$, where $\mathbb{V}_R(O,\rho)$ is the variance of the random variable~$\mathbf{X}_R$. Our first result characterizes this variance:
\begin{lemma}\label{lem:var_thrift}
The variance of the random variable~$\mathbf{X}_R$ is given by
\begin{align}\label{eq:var_thrift}
  \mathbb{V}_R(O, \rho)
&= \frac1R \mathbb{V}(O, \rho) + \frac{R-1}R \mathbb{V}_{\!*}(O, \rho),
\end{align}
where $\mathbb{V}(O, \rho)$ is the variance of the random variable $\mathbb{X}_1$, as in ordinary shadow estimation, while
\begin{equation}\label{eq:V4}
\begin{aligned}
  \mathbb{V}_{\!*}(O, \rho)
&:= \VV(\EE(\mathbf{X}_1|\mathbf{U})) \\
&\;= \mathbb{V}_U\bigl(\EE_x\tr\big(O\mathcal{F}^{-1} (U\ct \dens{x} U)\big)\bigr).
\end{aligned}
\end{equation}
\end{lemma}
For the sake of brevity we postpone the proof of this lemma and all following results to the Supplemental Material.
For~$R=1$ we recover the performance guarantee of ordinary shadow estimation \cite{huang2020predicting}.
To analyze the thrifty case, $R>1$, we need to estimate the term $\mathbb{V}_{\!*}(O, \rho)$, which depends on the fourth moment of the random circuits.
A straightforward corollary of \cref{lem:var_thrift} (using the law of total variance) is
\begin{align*}
  \mathbb{V}_R(O,\rho) \leq \mathbb{V}(O,\rho).
\end{align*}
However, this does \emph{not} imply that thrifty shadow estimation is always better than ordinary shadow estimation.
In fact the above argument allows for a range of possibilities, going from~$\mathbb{V}_R(O,\rho) \approx \mathbb{V}(O,\rho)/R$, in which case thrifty shadow estimation recovers the guarantees of ordinary shadow estimation for the same number of measurements (but might be preferable due to the lower cost of circuit reuse, as discussed in the introduction), to $\mathbb{V}_R(O,\rho) \approx \mathbb{V}(O,\rho)$, in which case setting~$R>1$ would be useless. We will see that both scenarios arise naturally when one performs thrifty shadow estimation with a unitary $4$ design or the multiqubit Clifford group respectively. We also give a parametrized family of circuit models that elegantly interpolates between these extremes.

\smallskip\textbf{Unitary 4-designs.}
We begin by analyzing the variance of thrifty shadow estimation for any circuit set that is a unitary 4-design. Our objective is to calculate \cref{eq:var_thrift}.
We are interested in the limit of many qubits, meaning we will be happy with estimates that include $\mathcal O(2^{-n})$ terms in all expressions. We obtain the following theorem.
\begin{theorem}\label{thm:4design}
The variance of thrifty shadow estimation with any 4-design circuit set satisfies
\begin{equation*}
\mathbb{V}_R(O,\rho) = \frac{1}{R} \mathbb{V}(O,\rho) + \frac{R-1}R \mathcal O(2^{-n}\tr(O^2))
\end{equation*}
for any traceless observable~$O$, with $\mathbb{V}(O, \rho)$ the variance associated with standard shadow estimation.
\end{theorem}
The proof of this theorem follows from Schur-Weyl duality for the unitary group, which is matched by any unitary $4$ design up to fourth order expressions (such as the variance).
A primer on Schur-Weyl duality can be found in the Supplemental Material.
We know that~$\mathbb{V}(O,\rho) \approx \tr(O^2)$, which means shadow estimation is scalable precisely when~$\tr(O^2)$ is polynomially bounded.
In this case \cref{thm:4design} tells us that~$\VV_R(O,\rho)$ and~$\VV(O,\rho)/R$ are exponentially close in the number of qubits -- in other words, circuit reuse essentially does not impact the statistical accuracy of shadow estimation. This means that thrifty shadow tomography, with access to a unitary $4$-design will in practice be preferable to standard shadow tomography. However, demanding access to an exact unitary $4$-design is a strong requirement, which we address in more detail later.

\smallskip\textbf{Multi-qubit Cliffords.}
Next we will \emph{lower bound} the variance of thrifty shadow estimation for the multi-qubit Clifford group.
In particular we will show that there are states $\rho$ and observables $O$ such that the variance $\VV_R(O,\rho)$ in \cref{eq:var_thrift} is \emph{independent} of $R$ (in the limit of many qubits).
Concretely:

\begin{theorem}\label{thm:clifford}
Consider thrifty shadow estimation with the $n$-qubit Clifford group~$\mathbb{C}_n$.
For any pure stabilizer state $\rho=\dens{S}$ and the traceless observable $O = \dens{S} - 2^{-n}I$, we have
\begin{equation*}
  \mathbb{V}_R(O,\rho) = 2+  \mathcal O(2^{-n}).
\end{equation*}
\end{theorem}
The proof of this theorem hinges on the recently developed Schur-Weyl duality theory for the Clifford group~\cite{gross2021schur} and can be found in the Supplemental Material.
There is thus a striking divergence in behavior between the Clifford group and and $4$-design when it comes to re-using circuits. This result formalizes an observation already made in experiment~\cite{struchalin2021experimental}, and serves as a warning for future experiments using this circuit set.

\smallskip\textbf{An interpolating family.}
From the preceding results one would prefer to perform thrifty shadow estimation with $4$-design circuits where circuit reuse is maximally useful.
Unfortunately no exact constructions of unitary $4$-designs are known (for an arbitrary number of qubits~$n$).
Moreover the Clifford group is not only useful due to its statistical properties, but also because it allows for the estimator~$\hat{o}_R$ to be computed efficiently in classical post-processing whenever the associated observable is a stabilizer state or a Pauli operator (or a well-behaved combination of these).
This is a property we would like to preserve as much as possible.
With this in mind, and inspired by~\cite{haferkamp2020quantum}, we consider a family of circuit sets that interpolates between the extreme cases discussed above.
Recall that the $\T$-gate is the non-Clifford unitary
\begin{align*}
  \T = \begin{pmatrix} 1& 0 \\ 0 & e^{i\pi/4}\end{pmatrix}.
\end{align*}
For a system of $n$ qubits, we denote by $\T$ this gate but acting on the first of $n$ qubits.
Then we can consider the following finite set of quantum circuits for any natural number~$k$:
\begin{equation}\label{eq:interpolating_set}
  \mathbb{U}_k = \braces*{ C_{k} \T C_{k-1} \cdots \T C_0 \;\;\;\;|\;\;\; C_0, \dots, C_k \in \mathbb{C}_n },
\end{equation}
This set is at least a $3$-design for any~$k$.
With increasing~$k$, it is an approximate $t$-design of any order~\cite{haferkamp2020quantum}.
Moreover, computing classically the overlap $\tr(O U\dens{x}U\ct)$~for stabilizer~$O$ and~$U\in \mathbb{U}_k$, which is required for computing the estimator~$\hat o$, can be achieved in time~$O(2^{0.396k})$~\cite{qassim2021improved}.

We now show that in the limit of large system sizes~$n$, the variance of thrifty shadow estimation with the circuit set~$\mathbb{U}_k$ approaches the result for 4-designs, which is $\mathbb{V}_R(O,\rho) \approx \mathbb{V}(O,\rho) / R$ for observables of bounded Hilbert-Schmidt norm (\cref{thm:4design}), up to an error that decreases exponentially with~$k$, leading to a classical simulation cost that is inverse  polynomial in the desired error.

\begin{theorem}\label{thm:interpolating}
The variance of thrifty shadow estimation with the circuit set~$\mathbb{U}_k$ defined in \cref{eq:interpolating_set} satisfies
\begin{align*}
\mathbb{V}_R(O,\rho) -\frac{1}{R} \mathbb{V}(O,\rho)  \!&\leq \!\frac{R\!-\!1}R \mathcal O(2^{-n}\tr(O^2))\\
&\hspace{1em}+ \frac{R-1}R 30 \tr(O^2) \parens*{ 1 + \mathcal O(2^{-n}) }\\
&\hspace{7em}\times \parens*{ \frac{3}{4} + \mathcal O(2^{-n}) }^k
\end{align*}
for any traceless observable~$O$, with $\mathbb{V}(O, \rho)$ the variance associated with standard shadow estimation.
\end{theorem}

While our result is inspired by~\cite{haferkamp2020quantum}, we do not know how to deduce it directly from their approximate 4-design result.
The reason for this is again that the support of the shadow estimation probability distribution~$p(U,x)$ grows as~$\mathcal O(2^n)$, and hence any additive error term will blow up correspondingly.

The advantage of thrifty shadow estimation with the interpolating family is best seen by considering a simple cost model.
Set $\tr(O^2)=1$ for simplicity, and assume that generating a new random circuit has cost~$\alpha\geq1$ and re-using it has unit cost.
Then we can express the cost~$\mathrm{C}$ for a total of~$N$~samples as $\mathrm{C} = (N/R) (\alpha + R - 1)$.
When using a median of means estimator, the  accuracy of thrifty (or standard for $R=1$) shadow estimation with $N$ samples is proportional to $\mathbb{V}_R/(N/R)$ (provided $N\geq K \ceil{8\log(1/\delta)}$, so that the estimator is well-defined).
We can express this in terms of the total cost~$\mathrm C$ as
\begin{equation}\label{eq:acc vs cost}
\frac {\mathbb{V}_R} {N/R} = \frac{\alpha+R-1}{\mathrm{C}} \mathbb{V}_R.
\end{equation}
This can be minimized to obtain the optimal number of repetitions~$R$ for a fixed random circuit generation cost~$\alpha$ and a fixed total cost budget~$\mathrm{C}$.
For the homeopathic circuit family with~$k$ $\T$-gates, using \cref{thm:interpolating} and taking the maximal possible value of $\mathbb{V}_R$, \cref{eq:acc vs cost} reads (suppressing small factors)
\begin{equation*}
\frac {\mathbb{V}_R} {N/R} \approx \frac{1}{\mathrm{C}} \bracks*{ \parens*{ \frac{\mathbb{V}_1}{R} + 30\,\frac{R-1}{R}\bigg(\frac{3}{4}\bigg)^k } (\alpha+R-1) },
\end{equation*}
leading to an optimal choice of~$R$ (every value of $R$ corresponds to a value of $N$ at fixed cost $\mathrm{C}$) given by
\begin{equation*}
  R \approx \sqrt{\frac{(1-\alpha)|(\mathbb{V}_1 -  30\, (3/4)^k )|}{30\, (3/4)^k}} .
\end{equation*}
This implies (if $\mathbb{V}_1\neq 0$ and $\alpha>1$) that for any value of~$\alpha$ and~$\mathrm{C}$, there is a value of~$k$ such that the optimal choice of~$R$ is~$R=N/K$.
 (accounting for the batching requirement in the median-of-means estimator), corresponding to a protocol where one samples a single circuit per batch and repeats it~$N/K$ times.
The computational cost of strongly simulating a quantum circuit with~$k$~many $\T$-gates currently~%
\footnote{This is a subject of active research, and thus this scaling might improve further} scales as $\mathcal O(2^{0.396k})$~\cite{qassim2021improved}.
Hence the optimal value of $R$ scales roughly with an inverse square $(2\log(3/4)/0.396 \approx -2.08)$ with the cost of simulation.
This means that thrifty shadow tomography with maximal reuse can be implemented using a circuit set that requires only polynomially more classical computational resources as compared (traditional) shadow tomography with the multi-qubit Clifford group.
We emphasize that this is a heuristic calculation, since we are ignoring some small terms in the expression for the homeopathic variance. In particular it is only accurate in the regime of many qubits (when the $O(2^{-n})$ terms are small).
However it shows that thrifty shadow estimation using the homeopathic interpolation circuit set can be a powerful alternative to standard shadow estimation when the cost of generating new circuits is high.

Finally, we note that sampling at least one circuit per batch is a requirement for the median-of-means estimator to function.
To see this, consider the extreme scenario where only one random circuit is sampled, and repeated many times, with the measurement outcomes grouped into batches as above.
As the number of repetitions increases, the median-of-means estimator will converge to the average value for the single random circuit, with the only remaining randomness due to circuit choice.
However, this randomness can still be ill-behaved (in the sense that the distribution is heavy-tailed), precluding exponential concentration of the estimator, which is required for shadow tomography.

\section{Tail bounds for shadow estimation}
In this section we revisit the use of median-of-means estimation in shadow estimation with circuit sets that are (at least) $3$-designs.
It has been noted before that for the circuit sets of single-qubit Clifford gates and matchgates, the median-of-means estimator can be replaced by the standard mean estimator without a loss in performance \cite{zhao2020fermionic}.
In this section we show that this is also true if the circuit set is the entire unitary group, but is not true if the circuit set is the multi-qubit Clifford group.
Hence, just like for thrifty shadow estimation, the Clifford group fails to fully emulate the statistical behavior of Haar random unitaries.

For simplicity we consider shadow estimation but both results also hold for thrifty shadow estimation with~$R>1$.
Throughout this section we will write $\mathbf{X}_n=\mathbf{X}$ to explicitly indicate the number of qubits~$n$ in a subscript.

\smallskip\textbf{Unitary group}
We first consider shadow estimation with the full unitary group. Somewhat surprisingly, we can prove sub-exponential behavior for shadow tomography with the standard mean estimator.

\begin{theorem}\label{thm:unit_tails}
Consider shadow estimation with the $n$-qubit unitary group as circuit set, state~$\rho$, and traceless observable~$O$.
For $N$ i.i.d.\ copies $\mathbf{X}^{(1)}_n, \ldots \mathbf{X}_n^{(N)}$ of~$\mathbf{X}_n$, we have a Bernstein-like tail~bound:
\begin{align*}
\Pr\mleft( \abs*{ \frac1N \sum_{i=1}^N \mathbf{X}_n^{(i)} - \mathbb{E}(\mathbf{X}_n) } \geq \eps \mright)
\leq
\begin{cases}
2 \exp\mleft(- \frac{N \eps^2} {48 \norm{O}_{\HS}^2} \mright) \\ \hspace{2em}\text{ if $\eps \leq 12 \norm{O}_{\HS}$}, \\
2 \exp\mleft(- \frac{N \eps} {4 \norm{O}_{\HS}} \mright) \\ \hspace{2em}\text{ if $\eps > 12 \norm{O}_{\HS}$}.
\end{cases}
\end{align*}
\end{theorem}
This theorem again follows from Schur-Weyl duality for the unitary group and a careful accounting of the moment generating function of $\mathbb{X}_n$.

\Cref{thm:unit_tails} shows that the median-of-means estimator in \cite{huang2020predicting} can in principle be replaced by a standard empirical average, as long as one uses the full unitary group as the circuit set.
This is akin to earlier such results for the single-qubit Clifford and matchgate groups.
However, these earlier statements were a consequence of the fact that the distributions being sampled from are bounded for all~$n$ (independently of~$n$, in terms of some separate locality parameter~$k$), and it is hence not surprising that an exponential tail bound can be established.
On the other hand, in the case of shadow estimation with the unitary group (or any $3$-design), the support of the distribution diverges as~$n\to \infty$, making such a statement significantly less trivial.
We believe that this exponential tail behavior is fundamentally a property of the full unitary group, making it difficult to achieve in practice.

\smallskip\textbf{Clifford group}
We will now argue that the opposite behavior holds when one averages over the multi-qubit Clifford group instead.
By this we mean that for shadow estimation with the Clifford group no ``useful'' tail bounds are possible.
This is a somewhat awkward statement to make, as for any finite number of qubits~$n$ the distribution associated with Clifford shadow estimation is bounded on the interval~$[-2^n, 2^n]$ (for all input states and observables), so it is always possible to obtain exponential tail bounds that grow exponentially in~$n$.
In \cref{thm:clifford_tails} below we show that, roughly speaking, one cannot do better (even if $\tr(O^2)$ is bounded).

\begin{theorem}\label{thm:clifford_tails}
Consider shadow estimation with the $n$-qubit Clifford group as circuit set, any $n$-qubit stabilizer state~$\rho=\dens{S}$, and the observable~$O = \dens{S} - 2^{-n} I$, so that~$\tr(O^2) \leq 1$.
Suppose that the random variables $\mathbf{X}_n$ satisfies a tail bound of the form
\begin{align}\label{eq:tail hypothesis}
  \Pr\mleft( \abs{ \mathbf{X}_n - \mathbb{E}(\mathbf{X}_n) } \geq t \mright) \leq A \exp\mleft( - \frac {t^\beta} {B_n} \mright),
\end{align}
for constants $A,\beta>0$ and a positive sequence $(B_n)$.
Then we have that $B_n = \tilde\Omega(2^{\beta n/4})$.
\end{theorem}
The key technical tool in this proof is a characterization of the $m$-th moments of $\mathbb{X}_n$, through the Schur-Weyl duality theory for the Clifford group. Concretely we obtain the following formula, which might be of independent interest:
\begin{align}\label{eq:m_th moment}
  \mathbb{E}(\mathbf{X}_n^m)
= (2^n+1)^m \sum_{k=0}^m \binom{m}{k} (-1)^{m-k} 2^{-n(m-k)} \prod_{\ell=0}^{k-1} \frac{2^\ell+1}{2^\ell + 2^n}.
\end{align}
One can see  that the moments~$\mathbb{E}(\mathbf{X}_n^m)$ of $\mathbf{X}_n$ grow fast with~$m$, and moreover increasingly so as~$n$ increases. This is key to the proof of \cref{thm:clifford_tails}
Note however that for fixed~$n$ the growth of the moments levels off when~$m\gg n$ (since~$\mathbf{X}_n$ is ultimately a bounded random variable).
Hence it is natural to discuss the behavior of the moments as we let~$n$ tend to infinity for fixed $m$.
In the case of the unitary group (as we saw in the proof of \cref{thm:unit_tails}), the limiting moments grow slowly enough with~$m$ to uniquely define a limiting random variable, with moments that are the limits of the moments of the random variables at finite~$n$.
However, this is not the case for the multi-qubit Clifford group. In the limit we have
\begin{equation*}
  \lim_{n\to\infty} \mathbb{E}(\mathbf{X}_n^m)
= \sum_{k=0}^m \binom{m}{k} (-1)^{m-k} \prod_{\ell=0}^{k-1}(2^\ell+1).
\end{equation*}
For~$m\geq 6$, the right-hand side this can be lower bounded as follows:
\begin{align*}
  \lim_{n\to\infty} \mathbb{E}(\mathbf{X}_n^m)
\geq 2^{\frac{m(m-1)}{2}}
= \Omega(2^{m^2/2})
\end{align*}
which shows that the moments grow super exponentially.
In fact they grow so fast that any random variable with those moments would have a moment generating function with convergence radius zero (and would hence be genuinely heavy tailed).
However the moments grow so fast it is not even clear whether the limiting moments determine a unique probability distribution.

\begin{acknowledgments}
\paragraph{Acknowledgments.}
We would like to thank Richard K\"ung for the discussion that kick-started this work.
We also thank Harold Nieuwboer and Freek Witteveen for useful conversations.
After a preprint of this article was posted, Ref.~\cite{zhou2022performance} reported on independent related results.
JH is partially supported by the Dutch Research Council (NWO) through the Gravitation grant Quantum Software Consortium, 024.003.037.
MW acknowledges support by the European Research Council~(ERC) through ERC Starting Grant 101040907-SYMOPTIC, the Deutsche Forschungsgemeinschaft (DFG, German Research Foundation) under Germany's Excellence Strategy - EXC\ 2092\ CASA - 390781972, the Federal Ministry of Education and Research (BMBF) through project Quantum Methods and Benchmarks for Resource Allocation (QuBRA), and NWO grant OCENW.KLEIN.267.
\end{acknowledgments}

\bibliography{variance}

\appendix
\begin{widetext}
\begin{center}
\Large \textsc{Supplemental Material}
\end{center}

\bigskip

In this supplement we provide detailed proofs of several statements made in the Letter.
We begin by collating several useful statements on the behavior of the unitary and multiqubit Clifford groups.

\section{Moments of the Unitary and the Clifford group}\label{sec:clifford_theory}
In this section we recall known results that allow for the computation of $t$-th moments of random unitaries, first in the very well-known setting of the unitary group and then for the Clifford group, where we draw on very recent results from the literature.

\subsection{Unitary group}\label{subsec:unitary}
We begin with a brief discussion of the moments of the Haar measure on the unitary group on~$n$ qubits.
These are captured collectively by the \emph{$t$-th moment (super)operator}, which is defined by
\begin{equation}\label{eq:moment op unitary}
  \mathcal M^{U(2^n),(t)} = \int_{U(2^n)} dU \,\mc{U}^{\ot t}.
\end{equation}
By standard arguments, $\mathcal M^{U(2^n),(t)}$ is the orthogonal projection onto the commutant of the $t$-fold tensor power action, that is, onto the linear space of operators on $((\CC^2)^{\ot n})^{\ot t}$ that commute with~$U^{\ot t}$ or, equivalently, are left invariant by~$\mathcal U^{\ot t}$ for any unitary~$U\in U(2^n)$.
By Schur-Weyl duality, this commutant is spanned by the natural action of the permutation group~$S_t$ on~$((\CC^2)^{\ot n})^{\ot t}$, that is, by the operators~$R_\pi$ for $\pi\in S_t$ permuting the $t$ copies of the $n$-qubit Hilbert space according to~$\pi$, i.e.,
\begin{align*}
  R_\pi = r_\pi^{\ot n}
\quad\text{where}\quad
  r_{\pi} = \sum_{x \in \{0,1\}^t} \ket{x_{\pi(1)},\ldots, x_{\pi(t)}} \bra{x_1, \ldots x_t}.
\end{align*}
For $n\geq t-1$, these operators are linearly independent and hence a basis, but not pairwise orthogonal.
However, for large~$n$ they are approximately orthogonal in the following sense.
Consider the \emph{Gram matrix}~$G^{U(2^n),(t)}$ associated with the basis~$\kett{R_\pi}$, which has as its entries
\begin{align*}
  G^{U(2^n),(t)}_{\pi,\pi'}
= \braakett{R_\pi}{R_{\pi'}}
= \braakett{r_\pi}{r_{\pi'}}^n
= 2^{c(\pi^{-1} \pi') n}
= \begin{cases}
2^{tn} & \text{ if } \pi = \pi', \\
\leq 2^{(t-1)n} & \text{ if } \pi \neq \pi'.
\end{cases}
\end{align*}
where $c(\tau)$ is the number of cycles in a permutation~$\tau\in S_t$.
Thus,
\begin{align}\label{eq:gram_pert U}
  G^{U(2^n),(t)} = 2^{tn} \parens*{ I + 2^{-n} E },
\end{align}
where~$E$ is a matrix with entries bounded by one.
It is in this sense that permutations become approximately orthogonal for a large number of qubits~$n$.

When the Gram matrix is invertible (as mentioned this is the case for $n\geq t-1$) then we call its inverse the \emph{Weingarten matrix}, denoted~$W^{U(2^n),(t)}$.
Then we can write the moment operator~\eqref{eq:moment op unitary} as
\begin{align}\label{eq:moment via weingarten unitary}
  \mathcal{M}_t = \sum_{\pi',\pi \in S_t} W^{U(2^n), (t)}_{\pi',\pi} \kett{R_{\pi'}}\braa{R_\pi}.
\end{align}
For large~$n$, one can show that \cref{eq:gram_pert U} implies that, for any fixed~$t$,
\begin{align}\label{eq:weingarten bound unitary}
  W^{U(2^n),(t)} = 2^{-tn} \parens*{ I + 2^{-n} F },
\end{align}
where~$F$ is a matrix with bounded entries.
We will not prove \cref{eq:weingarten bound unitary}, but note that it can be derived in the same fashion as the corresponding statement for the Clifford group, see \cref{eq:weingarten bound,eq:weingarten bound 2}, which we prove below.
The above is the starting point for an extensive theory known as \emph{Weingarten calculus}~\cite{collins2021weingarten}.

While in general the moment operator is nontrivial to compute with for finite~$n$, its action on the $t$-th tensor power of an $n$-qubit pure state~$\ket\psi$ is simple to compute and does not depend on the choice of state:
\begin{equation}\label{eq:state_av_haar}
 \mathcal M^{U(2^n),(t)}\kett{\psi}\tn{t}
= \int_{U(2^n)} dU \, \mc{U}^{\ot t} \kett{\psi}\tn t
= \int_{\mathbb{CP}^{2^n-1}} d\phi \, \kett{\phi}\tn{t}
= \parens*{ \prod_{\ell=0}^{t-1} \frac 1 {2^n+\ell} } \sum_{\pi\in S_t} \kett{R_\pi},
\end{equation}
where $\mathbb{CP}^{2^n-1}$ is the complex projective space of dimension $2^n-1$ (i.e., the space of quantum states).
This concludes our brief review of the moments of the unitary group.

\subsection{Clifford group}\label{subsec:cliff}
We now continue with the Clifford group~$\mathbb C_n$.
We will see that recent results in its representation theory imply a remarkably similar structure as what we discussed above for the unitary group.
We start by defining the \emph{$t$-moment operator} of the Clifford group, which similarly captures all $t$-th moments of random Clifford unitaries:
\begin{equation}\label{eq:moment op cliff}
  \mathcal M^{\mathbb C_n,(t)} = \frac{1}{\abs{\mathbb{C}_n}} \sum_{C \in\mathbb{C}_n} \mc{C}^{\ot t}.
\end{equation}
Similarly, $\mathcal M^{\mathbb C_n,(t)}$ is the orthogonal projection onto the commutant of the $t$-fold tensor power action of the Clifford group~$\mathbb C_n$, that is, onto the linear space of operators on $((\CC^2)^{\ot n})^{\ot t}$ that commute with~$C^{\ot t}$ or, equivalently, are left invariant by~$\mathcal C^{\ot t}$ for any Clifford unitary~$C\in\mathbb{C}_n$.

Because the Clifford group is a $3$-design, for $t\leq3$ this commutant coincides with that of the unitary group and is thus spanned by the action of the permutation group~$S_t$ (as above).
In particular, $\mathcal M^{\mathbb C_n,(t)} = \mathcal M^{U(2^n),(t)}$ for $t\leq 3$.
For~$t\geq4$, however, the commutant is more complicated, as the Clifford group is \emph{not} a $4$-design.
However, a series of works have started to uncover its structure~\cite{zhu2016clifford,helsen2018representations,gross2021schur,montealegre2021rank}.
In particular, the commutant has recently been characterized in full generality~\cite[Theorem~4.3]{gross2021schur}.
For~$n\geq t-1$, its dimension only depends on~$t$, and a basis is given by~$\kett{R_T}$ for $T\in\Sigma_{t,t}$, where $\Sigma_{t,t}$ is a certain set of $t$-dimensional linear subspaces $T \subseteq \mathbb{F}_2^{2t}$ and where, analogously to the above, the operators $R_T$ are defined by
\begin{align}\label{eq:def r_T}
  R_T = r_T^{\ot n}
\quad\text{where}\quad
  r_T = \sum_{(x,y) \in T} \ket{x_1,\dots,x_t}\bra{y_1,\dots,y_t}.
\end{align}
The set $\Sigma_{t,t}$ includes the subspaces $T_\pi = \{(\pi x, x) : x \in \mathbb{Z}_2^t\}$ associated with permutations~$\pi \in S_t$, with $R_{T_\pi} = R_\pi$.
By identifying $\pi$ with $T_\pi$, we shall thus think of $S_t \subseteq \Sigma_{t,t}$ as a subset.
For $t\leq 3$, $\Sigma_{t,t} = S_t$, but if~$t\geq 4$ there are other subspaces that do \emph{not} arise from permutations.

In the following we will mainly be interested in the case~$t=4$.
The set~$\Sigma_{4,4}$ consists of 30 subspaces and hence the commutant of the 4-th tensor power action of the Clifford group has dimension~$30$ for~$n\geq 3$.
It is spanned by 30 operators~$R_T$, which are described explicitly in~\cite[Example~4.27]{gross2021schur}.
Apart from the~24 permutation operators~$R_\pi$, parameterized by~$\pi \in S_4$, there are six more operators~$R_T$.
The latter can be written in the form~$R_\pi \Pi_4$, where $\pi$ ranges over the subgroup~$S_3 \subseteq S_4$ and
\begin{align}\label{eq:Pi_4}
  \Pi_4 := R_{T_4} = 2^{-n} \sum_{P \in \mathbb{P}_n} P^{\ot 4},
\end{align}
where $\mathbb{P}_n$ is the set of $n$-qubit Pauli operators and $T_4$ denotes a certain element in~$\Sigma_{4,4}$ whose details are not important here~%
\footnote{This follows from \cite[Example~4.27 and Theorem~4.24]{gross2021schur}.
Observe that the stabilizer of the left multiplication action of $S_4$ on $R_{T_4}$ is the Klein four-group~$K_4 = \langle (1 2)(3 4), (1 3)(2 4) \rangle$, with quotient $S_4 / K_4$ is isomorphic to $S_3$.}.
We denote by $\hat S_3 = \{ \pi T_4 \}_{\pi \in S_3}$ the set of subspaces corresponding to the six operators~$R_\pi \Pi_4$ for~$\pi\in S_3$.
Then, $\Sigma_{4,4} = S_4\cup \hat{S}_3$.

Next we consider the \emph{Gram matrix}~$G^{\mathbb C_n,(t)}$ associated with the basis~$\kett{R_T}$ for $T \in \Sigma_{t,t}$, which has as its entries
\begin{align*}
  G^{\mathbb C_n,(t)}_{T,T'}
= \braakett{R_T}{R_{T'}}
= \braakett{r_T}{r_{T'}}^n
= |T \cap T'|^n
= 2^{\dim(T \cap T') n}
= \begin{cases}
2^{tn} & \text{ if } T = T', \\
\leq 2^{(t-1)n} & \text{ if } T \neq T',
\end{cases}
\end{align*}
which follows from \cref{eq:def r_T}.
Accordingly, in complete analogy to \cref{eq:gram_pert U} we have
\begin{align}\label{eq:gram_pert cliff}
  G^{\mathbb C_n,(t)} = 2^{tn} \left( I + 2^{-n} E \right),
\end{align}
where $E$ is a matrix with entries bounded by one.
This also confirms the linear independence of the operators $R_T$ for large~$n$.
As mentioned above, it is known that $n\geq t-1$ suffices for linear independence, so $G^{\mathbb C_n,(t)}$ is invertible as soon as $n \geq t-1$.
We call the inverse of the Gram matrix the \emph{Clifford-Weingarten matrix}~$W^{\mathbb C_n,(t)}$.
Then the moment operator~\eqref{eq:moment op cliff} is given by
\begin{align}\label{eq:moment via weingarten clifford}
  \mathcal{M}^{\mathbb C_n,(t)} = \sum_{T',T \in \Sigma_{t,t}} W^{\mathbb C_n,(t)}_{T',T} \kett{R_{T'}}\braa{R_T}.
\end{align}

The Clifford-Weingarten matrix is nontrivial to compute with, but the situation simplifies significantly if it is applied to a tensor product of pure stabilizer density matrices $\kett{S}\tn{t}$  In this case we have (see \cite[Theorem~5.3]{gross2021schur})
\begin{equation}\label{eq:state_av_cliff}
  \mathcal M^{\mathbb C_n,(t)}\kett{S}\tn{t}
= \frac 1{2^n \prod_{\ell=0}^{t-2} (2^n + 2^\ell)} \sum_{T\in\Sigma_{t,t}} \kett{R_T}.
\end{equation}

For $t\leq 3$, this formula coincides with \cref{eq:state_av_haar} (for $\ket\psi=\ket S$).
This must be so, since the Clifford group is a unitary 3-design.

Finally, we note that $W^{\mathbb C_n,(t)}$ is diagonally dominant for large~$n$ in a similar manner as~$G^{\mathbb C_n,(t)}$ in \cref{eq:gram_pert cliff}.
We make this precise for~$t=4$, since only this case will be important for us (but a similar argument with different constants works for any~$t$).
The operator norm of the perturbation~$2^{-n} E$ in \cref{eq:gram_pert cliff} can be bounded as
\begin{align*}
  2^{-n} \norm E
&= 2^{-4n} \norm{G^{(4)} - 2^{4n}I } \\
&\leq 2^{-4n} \max_{T\in\Sigma_{4,4}} \sum_{T'\in\Sigma_{4,4}} \abs{G^{\mathbb C_n,(4)}_{T,T'} - 2^{4n} \delta_{T,T'} } \\
&= 2^{-4n} \max_{T\in\Sigma_{4,4}} \sum_{T'\in\Sigma_{4,4}, T' \neq T} \abs{G^{\mathbb C_n,(4)}_{T,T'}} \\
&= 2^{-4n} \parens*{ 7 \cdot 2^{3n} + 14 \cdot 2^{2n} + 8 \cdot 2^{n} },
\end{align*}
where we first used that the operator norm of a symmetric matrix can be upper bounded in terms of the maximum $\ell^1$-norm of its rows as a consequence of the estimate $\norm A \leq \sqrt{\norm A_{1\to1} \norm A_{\infty\to\infty}}$,which follows from H\"older's inequality, and then a direct calculation using the explicit description of~$\Sigma_{4,4}$ from above.
For $n\geq4$, we have $\norm E \leq 8$ and $2^{-n} \norm E \leq 1/2$.
Hence the inverse of $G^{\mathbb C_n,(4)} = 2^{4n} (I + 2^{-n} E)$, i.e., the Clifford-Weingarten matrix~$W^{\mathbb C_n,(4)}$, is given by a geometric series,
\begin{align}\label{eq:weingarten bound}
  W^{\mathbb C_n,(4)}
= 2^{-4n} \sum_{k=0}^\infty \parens*{ -2^{-n} E }^k
= 2^{-4n} \parens*{ I + 2^{-n} F },
\end{align}
where $F = \sum_{k=0}^\infty (-1)^k 2^{-kn} E^{k+1}$ is such that
\begin{align}\label{eq:weingarten bound 2}
  \norm F
\leq \norm E \sum_{k=0}^\infty \parens*{ 2^{-n} \norm E }^k
\leq \frac {\norm E} { 1 - 2^{-n} \norm E }
\leq 16.
\end{align}
\Cref{eq:weingarten bound,eq:weingarten bound 2} show that the off-diagonal entries of the Clifford-Weingarten matrix~$W^{\mathbb C_n,(4)}$ are again sub-leading by at least a factor~$\mathcal O(2^{-n})$ in the limit of large~$n$.
In particular, $W^{\mathbb C_n, (4)}$ is diagonally dominant for large $n$.

\section{Thrifty shadow tomography}
In this section we provide proofs of the statements made in the Letter regarding the statistical efficiency of thrifty shadow tomography.

\begin{lemma}[Reprint of \cref{lem:var_thrift}]\label{lem1}
The variance of the random variable~$\mathbf{X}_R$ is given by
\begin{align}
  \mathbb{V}_R(O, \rho)
&= \frac1R \mathbb{V}(O, \rho) + \frac{R-1}R \mathbb{V}_{\!*}(O, \rho),
\end{align}
where $\mathbb{V}(O, \rho)$ is the variance of the random variable $\mathbb{X}_1$, as in ordinary shadow estimation, while
\begin{align}
  \mathbb{V}_{\!*}(O, \rho)
:= \VV(\EE(\mathbf{X}_1|\mathbf{U}))
= \mathbb{V}_U\bigl(\EE_x\braa{O} \mathcal F^{-1} \mc{U}^\dagger \kett{x}\bigr).
\end{align}
\end{lemma}

\begin{proof}
By the law of total variance,
\begin{align}\label{eq:TVR}
  \VV_R(O, \rho)
= \VV(\mathbf{X}_R)
= \EE(\VV(\mathbf{X}_R|\mathbf{U})) + \VV(\EE(\mathbf{X}_R|\mathbf{U}))
= \frac1R \EE(\VV(\mathbf{X}_1|\mathbf{U})) + \VV(\EE(\mathbf{X}_1|\mathbf{U})),
\end{align}
where we used that $\mathbf{X}_R=\frac1R\sum_{r=1}^R \mathbf{X}^{(r)}$, where the $\mathbf{X}^{(r)}$ are identically distributed, and independent conditional on~$\mathbf{U}$.
On the other hand, again by the law of total variance,
\begin{align}\label{eq:TV1}
  \VV(O, \rho)
= \VV(\mathbf{X}_1)
= \EE(\VV(\mathbf{X}_1|\mathbf{U})) + \VV(\EE(\mathbf{X}_1|\mathbf{U})).
\end{align}
Together,
\begin{align*}
  \VV_R(O, \rho)
= \frac1R \left( \VV(O,\rho) - \VV(\EE(\mathbf{X}_1|\mathbf{U})) \right) + \VV(\EE(\mathbf{X}_1|\mathbf{U}))
= \frac1R \VV(O,\rho) + \frac{R-1}R \VV(\EE(\mathbf{X}_1|\mathbf{U})),
\end{align*}
and now the claim follows.
\end{proof}

\begin{theorem}[Reprint of \cref{thm:4design}]
The variance of thrifty shadow estimation with any 4-design circuit set satisfies
\begin{equation*}
\mathbb{V}_R(O,\rho) = \frac{1}{R} \mathbb{V}(O,\rho) + \frac{R-1}R \mathcal O(2^{-n}\tr(O^2))
\end{equation*}
for any traceless observable~$O$, with $\mathbb{V}(O, \rho)$ the variance associated with standard shadow estimation.
\end{theorem}
\begin{proof}
In view of equation $(3)$ in the main text it suffices to bound~$\mathbb{V}_{\!*}(O, \rho)$.
We begin by expanding equation $(4)$ in the main text,
\begin{align}
\nonumber
  \mathbb{V}_{\!*}(O, \rho)
&= \int_{U(2^n)} dU \parens*{ \sum_{x \in \{0,1\}^n} \braa{O} \mathcal F^{-1} \mc{U}^\dagger \kett{x} \braa{x} \mc{U} \kett{\rho} }^2
- (\tr O\rho)^2 \\
\label{eq:V4 expansion}
&= \underbrace{(2^n+1)^2 \int_{U(2^n)} dU \sum_{x\in\{0,1\}^n}\braa{x\tn{4}}\mc{U}\tn{4}\kett{(O\otimes \rho)\tn{2}}}_{(*)} \\
\nonumber
&+ \underbrace{(2^n+1)^2\int_{U(2^n)} dU\sum_{x\neq\hat{x}\in \{0,1\}^n}\braa{x\tn{2}\otimes \hat{x}\tn{2}}\mc{U}\tn{4}\kett{(O\otimes \rho)\tn{2}}}_{(**)} - \parens*{ \tr(O\rho) }^2,
\end{align}
where we used that $\mathbb{U}$ is a $4$-design to replace the average over $\mathbb{U}$ with a Haar average over the unitary group, as well as the expression of the inverse frame operator from \cite{huang2020predicting}.
We begin by analyzing the first term.
By \cref{eq:state_av_haar} we have exactly:
\begin{align*}
(*)
= (2^n+1)^2 \sum_{x\in \{0,1\}^n} \int_{U(2^n)} dU\braa{x\tn{4}}\mc{U}\tn{4}\kett{(O\otimes \rho)\tn{2}}
= \frac{2^n (2^n+1)^2}{2^n(2^n+1)(2^n+2)(2^n+3)} \sum_{\pi \in S_4} \braakett{R_\pi}{(O\otimes \rho)\tn{2}},
\end{align*}
where we used the invariance of the Haar measure to absorb the $2^n$ terms in the sum over $x\in \{0,1\}^n$
It is not hard to see (and we prove in \cref{lem:v_bound 1} in the appendix) that $\abs{\braakett{R_\pi}{(O\otimes \rho)\tn{2}}} \leq \tr(O^2)$ for any permutation~$\pi\in S_4$.
This immediately implies
\begin{align*}
(*) &\leq \frac{2^n+1}{(2^n+2)(2^n+3)} |S_4| \tr(O^2) = \mathcal O\bigl(\tr(O^2)2^{-n} \bigr).
\end{align*}
To bound the second term in \cref{eq:V4 expansion}, we first note that, for $x\neq\hat{x}$, we have $\braakett{x\tn{2}\otimes \hat{x}\tn{2}}{R_\tau } =0$ unless $\tau \in \{e, (12), (34), (12)(34)\}$, in which case~$\braakett{x\tn{2}\otimes \hat{x}\tn{2}}{R_\tau } = 1$.
We prove this in \cref{lem:v_bound 2} in the appendix.
Hence we have, with $W^{U(2^n), (4)}$ the Weingarten matrix associated to the unitary group, as defined above \cref{eq:moment via weingarten unitary},
\begin{align*}
(**) &= (2^n+1)^2 \int_{U(2^n)} dU \sum_{x\neq \hat{x}}\braa{x\tn{2}\otimes \hat{x}\tn{2}}\mc{U}\tn{4}\kett{(O\otimes \rho)\tn{2}}\\
&= (2^n+1)^2 \sum_{\tau,\pi\in S_4} W^{U(2^n),(4)}_{\tau, \pi} \sum_{x\neq \hat{x}}\braakett{x\tn{2}\otimes \hat{x}\tn{2}}{R_\tau} \braakett{R_\pi}{(O\otimes \rho)\tn{2}} \\
&= 2^n (2^n-1) (2^n+1)^2 \sum_{\pi \in S_4} \parens*{ W^{U(2^n),(4)}_{e, \pi} + W^{U(2^n),(4)}_{(12),\pi} + W^{U(2^n),(4)}_{(34),\pi}+ W^{U(2^n),(4)}_{(12)(34),\pi}} \braakett{R_\pi}{(O\otimes \rho)\tn{2}}.
\end{align*}
As above we know that $\abs{\braakett{R_\pi}{(O\otimes \rho)\tn{2}}} \leq \tr(O^2)$.
Combining this with the fact that $W^{U(2^n),(4)}_{\pi, \pi'} = 2^{-4n} (\delta_{\pi,\pi'} + \mathcal O(2^{-n}))$, as we know from \cref{eq:weingarten bound unitary}, we obtain
\begin{align*}
  (**)
&= \parens*{ \braakett{R_{(12)}}{(O\otimes \rho)\tn{2}} + \braakett{R_{(34)}}{(O\otimes \rho)\tn{2}}+ \braakett{R_{(12)(34)}}{(O\otimes \rho)\tn{2}} }
+ \mathcal O(2^{-n}\tr(O^2)) \\
&= \braakett{R_{(12)(34)}}{(O\otimes \rho)\tn{2}} + \mathcal O(2^{-n}\tr(O^2))
= \tr(O \rho)^2 + \mathcal O(2^{-n}\tr(O^2)),
\end{align*}
where we used that $\tr(O)=0$ by assumption.
If we insert the bounds on (*) and (**) into \cref{eq:V4 expansion} then we obtain
\begin{equation}\label{eq:V4 unitary}
\mathbb{V}_{\!*}(O,\rho)
= (\tr O\rho)^2 + \mathcal O(2^{-n}\tr(O^2)) - (\tr O\rho)^2
= \mathcal O(2^{-n}\tr(O^2)).
\end{equation}
By equation $(3)$ from the main text, and noting that $\frac{R-1}{R}\leq 1$, this implies the desired result.
\end{proof}



\begin{theorem}[Reprint of \cref{thm:clifford}]
Consider thrifty shadow estimation with the $n$-qubit Clifford group~$\mathbb{C}_n$.
For any pure stabilizer state $\rho=\dens{S}$ and the traceless observable $O = \dens{S} - 2^{-n}I$, we have
\begin{equation*}
  \mathbb{V}_R(O,\rho) = 2+  \mathcal O(2^{-n}).
\end{equation*}
\end{theorem}
\begin{proof}
We again start from equation $(3)$ in the main text.
Using our choice of~$\rho$ and~$O$, one finds that
\begin{align*}
  \mathbb{V}(O, \rho)
&= \frac{2^n+1}{2^n+2} \parens*{ \tr(O^2) + 2 \tr(\rho O^2) } - \parens*{ \tr(\rho O) }^2 \\
&= \frac{2^n+1}{2^n+2} \parens*{ \tr\mleft( \parens*{ \dens{S} - 2^{-n}I }^2 \mright) + 2 \tr\mleft( \dens{S} \parens*{ \dens{S} - 2^{-n}I }^2 \mright) } - \parens*{ \tr\mleft( \dens{S} \parens*{ \dens{S} - 2^{-n}I } \mright) }^2 \\
&= 2 + \mathcal O(2^{-n}),
\end{align*}
and thus
\begin{equation}\label{eq:V cliff}
\mathbb{V}_R(O,\rho) = \frac{2}{R}+ \frac{R-1}{R}\mathbb{V}_{\!*}(O,\rho) + \mathcal O(2^{-n}).
\end{equation}
In the remainder of the proof we will show that $\mathbb{V}_{\!*} = 2 + \mathcal O(2^{-n})$, which suffices to establish the result.
The strategy is to relate~$\mathbb{V}_{\!*}$ to the corresponding quantity \emph{for the unitary group}, which we know to be small from \cref{eq:V4 unitary}.
Similarly to the derivation of \cref{eq:V4 expansion},
\begin{align}\label{eq:Vstar cliff}
  \mathbb{V}_{\!*}(O, \rho)
&= \frac{(2^n+1)^2}{\abs{\mathbb{C}_n}} \sum_{C \in\mathbb{C}_n} \sum_{x, \hat x \in \{0,1\}^n} \braa{x^{\ot 2} \ot \hat x^{\ot 2}} \mc{C}^{\ot 4} \kett{(O \ot \rho)^{\ot 2}} - \parens*{ \Tr(\rho O) }^2.
\end{align}
We would like to compare this with the analogous expression for the unitary group, which we know from \cref{eq:V4 expansion,eq:V4 unitary} satisfies
\begin{align}\label{eq:Vstar unitary}
  (2^n+1)^2 \int_{U(2^n)} dU \sum_{x, \hat x \in \{0,1\}^n} \braa{x^{\ot 2} \ot \hat x^{\ot 2}} \mc{U}^{\ot 4} \kett{(O \ot \rho)^{\ot 2}} - \parens*{ \Tr(\rho O) }^2
= \mathcal O(2^{-n}),
\end{align}
since $\tr(O^2) = 1 - 2^{-n}$ by our choice of~$O = \dens{S} - 2^{-n}I$.
To do this, we note that by our choice of $\rho=\dens{S}$ and~$O$, we have
\begin{align*}
  \mathcal U^{\ot 4} \kett{(O \ot \rho)^{\ot 2}}
= \mathcal U^{\ot 4} \kett{\rho^{\ot 4}}
- 2^{-n} \kett{I} \ot \mathcal U^{\ot 3} \kett{\rho^{\ot 3}}
- 2^{-n} \mathcal U^{\ot 2} \kett{\rho^{\ot 2}} \ot \kett{I} \ot \mathcal U \kett{\rho}
+ 4^{-n} \kett{I} \ot \mathcal U \kett{\rho} \ot \kett{I} \ot \mathcal U \kett{\rho}.
\end{align*}
For the last three terms, averaging the above over the unitary group gives the same result if we only average over the Clifford group, as the latter is a $3$-design.
For the first term we have, by \cref{eq:state_av_haar,eq:state_av_cliff}, and noting that $\rho=\kett S$ is a stabilizer state:
\begin{align*}
  \frac1{\abs{\mathbb{C}_n}} \sum_{C \in\mathbb{C}_n} \mc{C}^{\ot 4} \kett{\rho^{\ot 4}}
&= \frac 1{2^n (2^n + 1) (2^n + 2) (2^n+4)} \parens*{ \sum_{\pi\in S_4} \kett{R_\pi} + \sum_{T\in \hat{S}_3} \kett{R_T} }, \\
  \int_{U(2^n)} dU \; \mc{U}^{\ot 4} \kett{\rho^{\ot 4}}
&= \frac 1{2^n (2^n + 1) (2^n + 2) (2^n+3)} \sum_{\pi\in S_4} \kett{R_\pi},
\end{align*}
where we recall that $\Sigma_{4,4} = S_4 \cup \hat{S}_3$.
It follows from this and \cref{eq:Vstar cliff,eq:Vstar unitary} that
\begin{align}
\nonumber
  \mathbb{V}_{\!*}(O, \rho)
&= \frac{(2^n+1)^2}{\abs{\mathbb{C}_n}} \sum_{C \in\mathbb{C}_n} \sum_{x, \hat x \in \{0,1\}^n} \braa{x^{\ot 2} \ot \hat x^{\ot 2}} \mc{C}^{\ot 4} \kett{\rho^{\ot 4}} \\
\nonumber
&\quad\hspace{5em} - (2^n+1)^2 \int_{U(2^n)} dU \sum_{x, \hat x \in \{0,1\}^n} \braa{x^{\ot 2} \ot \hat x^{\ot 2}} \mc{U}^{\ot 4} \kett{\rho^{\ot 4}}
+ \mathcal O(2^{-n}) \\
\label{eq:Vstar cliff 2}
&= \underbrace{\frac {(2^n+1)^2} {2^n (2^n + 1) (2^n + 2) (2^n+4)} \sum_{x, \hat x \in \{0,1\}^n} \sum_{T\in \hat{S}_3} \braakett{x^{\ot 2} \ot \hat x^{\ot 2}}{R_T}}_{(*)} \\
\nonumber
&\quad \hspace{2em}- \underbrace{\frac {(2^n+1)^2} {2^n (2^n + 1) (2^n + 2) (2^n+3)} \parens*{1 - \frac{2^n+3}{2^n+4} } \sum_{x, \hat x \in \{0,1\}^n} \sum_{\pi\in S_4} \braakett{x^{\ot 2} \ot \hat x^{\ot 2}}{R_\pi}}_{(**)}
\;+\; \mathcal O(2^{-n}).
\end{align}
Finally, note that
\begin{align*}
  (*)
&= \frac {(2^n+1)^2} {2^n (2^n + 1) (2^n + 2) (2^n+4)} \sum_{x, \hat x \in \{0,1\}^n} \sum_{T\in \hat{S}_3} \braakett{x^{\ot 2} \ot \hat x^{\ot 2}}{R_T} \\
&= \frac {(2^n+1)^2} {2^n (2^n + 1) (2^n + 2) (2^n+4)} \parens*{ 2^n \cdot 6 + 2^n (2^n-1) \cdot 2 } \\
&= 2 + \mathcal O(2^{-n}),
\end{align*}
using the explicit formula for $\braakett{x^{\ot 2} \ot \hat x^{\ot 2}}{R_T}$ from \cref{lem:v_bound 2} in the appendix.
Furthermore,
\begin{align*}
  (**)
= \frac {(2^n+1)^2} {2^n (2^n + 1) (2^n + 2) (2^n+3) (2^n+4)} \sum_{x, \hat x \in \{0,1\}^n} \sum_{\pi\in S_4} \braakett{x^{\ot 2} \ot \hat x^{\ot 2}}{R_\pi}
\leq \mathcal O(2^{-n}),
\end{align*}
where we used that $\abs{\braakett{x^{\ot 2} \ot \hat x^{\ot 2}}{R_\pi} } \leq 1$ by H\"older's inequality for the Schatten norm.
If we substitute the above into \cref{eq:Vstar cliff 2} and the latter in turn into \cref{eq:V cliff} we obtain the desired result.
\end{proof}

\begin{theorem}[Reprint of \cref{thm:interpolating}]
The variance of thrifty shadow estimation with the circuit set~$\mathbb{U}_k$ defined in equation $(5)$ in the main text satisfies
\begin{equation*}
 0\leq\mathbb{V}_R(O,\rho) -\frac{1}{R} \mathbb{V}(O,\rho)  \leq \frac{R-1}R \mathcal O(2^{-n}\tr(O^2))
+ \frac{R-1}R 30 \tr(O^2) \parens*{ 1 + \mathcal O(2^{-n}) } \parens*{ \frac{3}{4} + \mathcal O(2^{-n}) }^k
\end{equation*}
for any traceless observable~$O$, with $\mathbb{V}(O, \rho)$ the variance associated with standard shadow estimation.
\end{theorem}
\begin{proof}
Recall from \cref{eq:moment op unitary,eq:moment op cliff} that the moment operators for the multi-qubit unitary and Clifford groups are defined as
\begin{equation*}
  \mc{M}^{U(2^n),(4)} = \int_{U(2^n)} dU\, \mc{U} \tn{4}
\quad\text{and}\quad
  \mc{M}^{\mathbb C_n,(4)} = \frac1{\abs{\mathbb{C}_n}} \sum_{C\in \mathbb{C}_n}\mc{C}\tn{4}.
\end{equation*}
Following our usual notation, $\mathcal T$ denotes the superoperator applying the $\T$-gate on the first qubit.
Then the moment operator associated to the circuit set~$\mathbb{U}_k$ is given by
\begin{align}\label{eq:homeo moment}
  \frac1{\mathbb{U}_k} \sum_{U \in \mathbb{U}_k} \mathcal U^{\ot4}
= \mc{M}^{\mathbb C_n,(4)} \mathcal T^{\ot4} \mc{M}^{\mathbb C_n,(4)} \cdots \mc{M}^{\mathbb C_n,(4)} \mathcal T^{\ot4} \mc{M}^{\mathbb C_n,(4)}
= \parens*{ \mc{M}^{\mathbb C_n,(4)} \mathcal T^{\ot4} }^k \mc{M}^{\mathbb C_n,(4)}.
\end{align}
By the unitary invariance of the Haar measure, we have
\begin{align*}
  \mc{T}^{\ot4} \mc{M}^{U(2^n),(4)} = \mc{M}^{U(2^n),(4)} \mc{T}^{\ot4} = \mc{M}^{U(2^n),(4)},
\end{align*}
as well as
\begin{align*}
  \mc{M}^{\mathbb C_n,(4)} \mc{M}^{U(2^n),(4)} = \mc{M}^{U(2^n),(4)} \mc{M}^{\mathbb C_n,(4)} = \mc{M}^{U(2^n),(4)}.
\end{align*}
Accordingly, we can rewrite \cref{eq:homeo moment} as follows:
\begin{align}\label{eq:uk moment}
  \frac1{\mathbb{U}_k} \sum_{U \in \mathbb{U}_k} \mathcal U^{\ot4}
= \mc{M}^{U(2^n),(4)} - \mc{E},
\quad\text{where}\quad
\mc{E} := \bracks*{ \parens*{ \mc{M}^{\mathbb C_n,(4)} - \mc{M}^{U(2^n),(4)} } \mathcal T^{\ot4}  }^k \parens*{ \mc{M}^{\mathbb C_n,(4)} - \mc{M}^{U(2^n),(4)} }
\end{align}
Using \cref{eq:moment via weingarten unitary,eq:moment via weingarten clifford}, we can express~$\mc{M}^{\mathbb C_n,(4)}-\mc{M}^{U(2^n),(4)}$ in the 30-dimensional non-orthogonal basis~$\kett{R_T}$ for~$T\in\Sigma_{4,4}$,
\begin{align*}
  \mc{M}^{\mathbb C_n,(4)} -\mc{M}^{U(2^n),(4)}
= \sum_{T, T'\in \Sigma_{4,4}} P_{T,T'} \kett{R_T}\braa{R_{T'}},
\end{align*}
where
\begin{align*}
  P = \bracks*{ \begin{array}{c|c}
    \quad W^{\mathbb{C}_n, (t)}\big|_{S_4 \times S_4} - W^{U(2^n), (t)} \quad & \quad W^{\mathbb{C}_n, (t)}\big|_{S_4 \times \hat{S}_3} \quad \\[.2cm]
    \hline\\[-0.4cm]
    W^{\mathbb{C}_n, (t)}\big|_{\hat{S}_3 \times S_4} & \quad W^{\mathbb{C}_n, (t)}\big|_{\hat{S}_3 \times \hat{S}_3} \quad
  \end{array} }
\end{align*}
is a $(24+6)\times(24+6)$ block matrix with respect to $\Sigma_{4,4} = S_4 \cup \hat{S}_3$, with $W^{U(2^n), (t)}$ the $24\times24$ Weingarten matrix of the unitary group and $W^{\mathbb{C}_n, (t)}$ the Weingarten matrix for the Clifford group.
From \cref{eq:weingarten bound unitary,eq:weingarten bound}, we know that $W^{U(2^n),(4)} = 2^{-4n} \parens*{ I_{24} + 2^{-n} F' }$ and $W^{\mathbb{C}_n,(4)} = 2^{-4n} \parens*{ I_{30} + 2^{-n} F'' }$ for matrices $F'$, $F''$ with bounded entries.
Thus the block matrix $P$ takes the form
\begin{align}\label{eq:P}
  P = 2^{-4n} \bracks*{ \begin{array}{c|c}
     0  &  0  \\
    \hline
    0 &  I
  \end{array} }
  + \mathcal O(2^{-5n}).
\end{align}
Next, consider a matrix~$Q$ with entries
\begin{align*}
  Q_{T,T'} := \braa{R_T} \mc{T}\tn{4} \kett{R_{T'}}
\end{align*}
for $T,T'\in\Sigma_{4,4}$.
Because $\mc{T}^{\ot 4}$ commutes with permutations, whenever $T\in S_4$ or $T'\in S_4$, or both, we have, by \cref{eq:gram_pert cliff}, that
\begin{align*}
  Q_{T,T'}
= \braakett{R_T}{R_{T'}}
= G^{\mathbb C_n,(4)}_{T,T'}
= 2^{4n} \parens*{ \delta_{T,T'} + \mathcal O(2^{-n}) },
\end{align*}
while if both $T\in\hat{S}_3$ and $T'\in\hat{S}_3$ then we prove in \cref{lem:t_gate} that
\begin{align*}
  Q_{T,T'}
= 2^{4n} \parens*{ \frac34 \delta_{T,T'} + \mathcal O(2^{-n}) }.
\end{align*}
Together,
\begin{align*}
  Q = 2^{4n} \bracks*{ \begin{array}{c|c}
     I  &  0  \\
    \hline
    0 &  \frac34I
  \end{array} }
  + \mathcal O(2^{3n}),
\end{align*}
and hence
\begin{align}\label{eq:PQ}
  P Q
= \bracks*{ \begin{array}{c|c}
     0  &  0  \\
    \hline
    0 & \frac34 I
  \end{array} }
  + \mathcal O(2^{-n}).
\end{align}
We are finally in a position to compare
\begin{align*}
  \mathbb{V}_{\!*}(O, \rho)
&= \frac{(2^n+1)^2}{\abs{\mathbb{U}_k}} \sum_{U \in\mathbb{U}_k} \sum_{x, \hat x \in \{0,1\}^n} \braa{x^{\ot 2} \ot \hat x^{\ot 2}} \mc{U}^{\ot 4} \kett{(O \ot \rho)^{\ot 2}} - \parens*{ \Tr(\rho O) }^2.
\end{align*}
with the analogous expression for the unitary group, which we know from \cref{eq:V4 expansion,eq:V4 unitary} satisfies
\begin{align*}
  (2^n+1)^2 \sum_{x, \hat x \in \{0,1\}^n} \braa{x^{\ot 2} \ot \hat x^{\ot 2}} \mc{M}^{U(2^n),(4)} \kett{(O \ot \rho)^{\ot 2}} - \parens*{ \Tr(\rho O) }^2
= \mathcal O(2^{-n}\tr(O^2)).
\end{align*}
In view of \cref{eq:uk moment}, it follows that
\begin{align*}
  \mathbb{V}_{\!*}(O, \rho)
=
- (2^n+1)^2 \sum_{x, \hat x \in \{0,1\}^n} \braa{x^{\ot 2} \ot \hat x^{\ot 2}} \mc{E} \kett{(O \ot \rho)^{\ot 2}}
+ \mathcal O(2^{-n}\tr(O^2)),
\end{align*}
where
\begin{align*}
  \mc{E}
= \bracks*{ \parens*{ \mc{M}^{\mathbb C_n,(4)} - \mc{M}^{U(2^n),(4)} } \mathcal T^{\ot4}  }^k \parens*{ \mc{M}^{\mathbb C_n,(4)} - \mc{M}^{U(2^n),(4)} }
= \sum_{T, T'\in \Sigma_{4,4}} \parens*{ (P Q)^k P }_{T,T'} \kett{R_T}\braa{R_{T'}}.
\end{align*}
Denote by $v_\psi$ the vector with entries $(v_\psi)_T := \braakett{\psi}{R_T}$ for $T\in\Sigma_{4,4}$.
Then,
\begin{align*}
  (2^n+1)^2 \sum_{x, \hat x \in \{0,1\}^n} \braa{x^{\ot 2} \ot \hat x^{\ot 2}} \mc{E} \kett{(O \ot \rho)^{\ot 2}}
&= (2^n+1)^2 \sum_{x, \hat x \in \{0,1\}^n} \bra{v_{x^{\ot 2} \ot \hat x^{\ot 2}}} (P Q)^k P \ket{v_{(O \ot \rho)^{\ot 2}}} \\
&\leq (2^n+1)^2 \sum_{x, \hat x \in \{0,1\}^n} \norm{v_{x^{\ot 2} \ot \hat x^{\ot 2}}} \norm{(P Q)^k P} \norm{v_{(O \ot \rho)^{\ot 2}}} \\
&\leq 30 \frac {(2^n+1)^2 2^{2n}} {2^{4n}} \tr(O^2) \, \parens*{ 1 + \mathcal O(2^{-n}) } \parens*{ \frac34 + \mathcal O(2^{-n}) }^k \\
&= 30 \tr(O^2) \parens*{ 1 + \mathcal O(2^{-n}) } \parens*{ \frac{3}{4} + \mathcal O(2^{-n}) }^k,
\end{align*}
where in the last inequality we used that $\norm{v_{x^{\ot 2} \ot \hat x^{\ot 2}}} \leq \sqrt{30}$ by \cref{lem:v_bound 2}, $\norm{(P Q)^k P} \leq (3/4 + \mathcal O(2^{-n}))^k 2^{-4n} (1 + \mathcal O(2^{-n}))$ by \cref{eq:P,eq:PQ}, and $\norm{v_{(O \ot \rho)^{\ot 2}}} \leq \sqrt{30}\tr(O^2)$ by \cref{lem:v_bound 1}.
We conclude that
\begin{align*}
  \mathbb{V}_{\!*}(O, \rho)
\leq 30 \tr(O^2) \parens*{ 1 + \mathcal O(2^{-n}) } \parens*{ \frac{3}{4} + \mathcal O(2^{-n}) }^k + \mathcal O(2^{-n}\tr(O^2)).
\end{align*}
Now the claim follows from equation $(3)$ in the main text.
\end{proof}

\section{Tail bounds for shadow estimation}\label{sec:tails}
In this section we revisit the use of median-of-means estimation in shadow estimation with circuit sets that are (at least) $3$-designs. We will prove \cref{thm:unit_tails,thm:clifford_tails}.

\begin{theorem}[Reprint of \cref{thm:unit_tails}]
Consider shadow estimation with the $n$-qubit unitary group as circuit set, state~$\rho$, and traceless observable~$O$.
The moment generating function of the random variable~$\mathbf{X}_n=\mathbf{X}$, for $\abs{t} < \norm{O}_{\HS}^{-1}$, upper bounded as
\begin{equation*}
  \mathbb{E}\mleft( e^{t \mathbf{X}_n} \mright)
\leq 1 + t\tr(O\rho) + t^2 \norm{O}_{\HS}^2 \frac{3- 2 t\norm{O}_{\HS}}{(1- t\norm{O}_{\HS})^2}.
\end{equation*}
Moreover, for $N$ i.i.d.\ copies $\mathbf{X}^{(1)}_n, \ldots \mathbf{X}_n^{(N)}$ of~$\mathbf{X}_n$, we have a Bernstein-like tail~bound:
\begin{align*}
\Pr\mleft( \abs*{ \frac1N \sum_{i=1}^N \mathbf{X}_n^{(i)} - \mathbb{E}(\mathbf{X}_n) } \geq \eps \mright)
\leq
\begin{cases}
2 \exp\mleft(- \frac{N \eps^2} {48 \norm{O}_{\HS}^2} \mright) & \text{ if $\eps \leq 12 \norm{O}_{\HS}$}, \\
2 \exp\mleft(- \frac{N \eps} {4 \norm{O}_{\HS}} \mright) & \text{ if $\eps > 12 \norm{O}_{\HS}$}.
\end{cases}
\end{align*}
\end{theorem}
\begin{proof}
We wish to compute the moment generating function
\begin{align*}
  \mathbb{E}\mleft( e^{t \mathbf{X}_n} \mright)
= \int_{U(2^n)} dU \sum_{x\in\{0,1\}^n} \braa{x} \mc{U} \kett{\rho} e^{t \braa{O}\mathcal{F}^{-1} \mc{U}\ct\kett{x}}
= \int_{U(2^n)} dU \sum_{x\in\{0,1\}^n} \braa{x} \mc{U} \kett{\rho} e^{t (2^n+1) \braa{x} \mc{U} \kett{O}}.
\end{align*}
Using the series expansion of the exponential function, we get
\begin{align}
\nonumber
  \mathbb{E}\mleft( e^{t \mathbf{X}_n} \mright)
&= \int_{U(2^n)} dU \sum_{x\in\{0,1\}^n} \braa{x} \mc{U} \kett{\rho} \sum_{m=0}^\infty \frac { \parens*{ t (2^n+1) \braa{x} \mc{U} \kett{O} }^m } {m!} \\
\nonumber
&= \sum_{m=0}^\infty \frac {t^m (2^n+1)^m} {m!} \sum_{x\in\{0,1\}^n} \int_{U(2^n)} dU \braa{x^{\ot(m+1)}} \mc{U}^{\ot(m+1)} \kett{O^{\ot m} \ot \rho} \\
\label{eq:series}
&= \sum_{m=0}^\infty \frac {t^m 2^n (2^n+1)^m} {m!} \parens*{ \prod_{\ell=0}^m \frac 1 {2^n+\ell} } \sum_{\pi\in S_{m+1}} \braakett{R_\pi}{O^{\ot m} \ot \rho},
\end{align}
where we used \cref{eq:state_av_haar} in the last step to evaluate the Haar integral.
Note that for any finite number of qubits~$n$, this series converges for all values of~$t$.
We wish to prove an $n$-independent upper bound.
To start,
\begin{align*}
  \braakett{R_\pi}{O^{\ot m} \ot \rho}
= \tr\mleft( R_\pi^\dagger \parens*{ O^{\ot m} \ot \rho } \mright)
= \tr(\rho O^{k_1}) \tr(O^{k_2}) \cdots \tr(O^{k_j})
\end{align*}
for certain numbers $k_1 + \dots + k_j = m$ and $j\geq1$ that depend on the disjoint cycle decomposition of the permutation~$\pi$.
Now,
\begin{align*}
  \abs*{ \tr(\rho O^k) } \leq \norm{O^k}_\infty = \norm{O}^k_\infty \leq \norm{O}^k_{\HS}
\quad\text{and}\quad
  \tr(O^{k+1}) \leq \norm{O}_{\HS} \norm{O^k}_{\HS} \leq \norm{O}_{\HS}^{k+1},
\end{align*}
for all $k\geq0$, while $\tr(O)=0$, and hence
$\abs*{ \braakett{R_\pi}{O^{\ot m} \ot \rho} } \leq \norm{O}_{\HS}^m$.
We can thus upper bound \cref{eq:series}, as follows:
\begin{align*}
  \mathbb{E}\mleft( e^{t \mathbf{X}_n} \mright)
&\leq 1 + t \tr(O \rho) + \sum_{m=2}^\infty \frac {(m+1) 2^n (2^n+1)^m} {\prod_{\ell=2}^m (2^n+\ell) } \parens*{ \abs{t} \norm{O}_{\HS} }^m \\
&\leq 1 + t \tr(O \rho) + \sum_{m=2}^\infty (m+1) \parens*{ \abs{t} \norm{O}_{\HS} }^m \\
&= 1 + t \tr(O \rho) + t^2 \norm{O}_{\HS}^2 \frac {3 - 2 \abs{t} \norm{O}_{\HS}} {\parens*{ 1 - \abs{t} \norm{O}_{\HS} }^2},
\end{align*}
where we used that $\sum_{m=0}^\infty (m+1) a^m = 1 / (1-a)^2$ for $\abs{a}<1$, with $a = t\norm{O}_{\HS}$.
This proves the first claim of the theorem.

If we further restrict to $\abs{t} \leq \norm{O}_{\HS}^{-1}/2$, then we see from $1+b\leq \exp(b)$ that
\begin{align*}
  \mathbb{E}\mleft( e^{t \mathbf{X}_n} \mright)
\leq 1 + t \tr(O \rho) + 12 t^2 \norm{O}_{\HS}^2
\leq e^{t \tr(O \rho) + 12 t^2 \norm{O}_{\HS}^2}.
\end{align*}
By Markov's inequality,
\begin{align*}
  \Pr\mleft( \mathbf{X}_n - \Tr(O \rho) \geq \eps \mright)
= \Pr\mleft( e^{t\mathbf{X}_n} \geq e^{t\eps + t\Tr(O \rho)} \mright)
\leq \frac {\mathbb{E}\mleft( e^{t \mathbf{X}_n} \mright)} {e^{t\eps + t\Tr(O \rho)}}
= e^{12 t^2 \norm{O}_{\HS}^2 - t\eps}.
\end{align*}
The choice for~$t$ that minimizes the right-hand side is~$t = \eps\norm{O}_{\HS}^{-2}/24$ if~$\eps \leq 12\norm{O}_{\HS}$, and otherwise~$t = \norm{O}_{\HS}^{-1}/2$, leading to the right tail bound
\begin{align*}
  \Pr\mleft( \mathbf{X}_n - \Tr(O \rho) \geq \eps \mright)
\leq \begin{cases}
\exp\mleft(- \frac{\eps^2} {48 \norm{O}_{\HS}^2} \mright) & \text{ if $\eps \leq 12 \norm{O}_{\HS}$}, \\
\exp\mleft(3 - \frac{\eps} {2 \norm{O}_{\HS}} \mright) \leq \exp\mleft(- \frac{\eps} {4 \norm{O}_{\HS}} \mright) & \text{ if $\eps > 12 \norm{O}_{\HS}$}.
\end{cases}
\end{align*}
We can similarly get a left tail bound, and consequently a two-sided bound by the union bound.
The second claim of the theorem then follows by noting that the moment generating function of a sum of i.i.d.\ random variables factors.
\end{proof}

\begin{lemma}[Proving equation~\eqref{eq:m_th moment}]
Consider shadow estimation with the $n$-qubit Clifford group as circuit set, any $n$-qubit stabilizer state~$\rho=\dens{S}$, and the observable~$O = \dens{S} - 2^{-n} I$.
Let~$\mathbf{X}_n=\mathbf{X}$ denote the associated random variable.
Then, for all~$m\geq0$,
\begin{align}\label{eq:m_th moment appendix}
  \mathbb{E}(\mathbf{X}_n^m)
= (2^n+1)^m \sum_{k=0}^m \binom{m}{k} (-1)^{m-k} 2^{-n(m-k)} \prod_{\ell=0}^{k-1} \frac{2^\ell+1}{2^\ell + 2^n}.
\end{align}
\end{lemma}
\begin{proof}
We have
\begin{align*}
  \mathbb{E}(\mathbf{X}_n^m)
&= \frac{(2^n+1)^m}{\abs{\mathbb{C}_n}} \sum_{C\in\mathbb{C}_n} \sum_{x\in\{0,1\}^n} \braa{x^{\ot(m+1)}} \mc{C}^{\ot(m+1)} \kett{\rho \ot O^{\ot m}} \\
&= \frac{(2^n+1)^m}{\abs{\mathbb{C}_n}} \sum_{C\in\mathbb{C}_n} \sum_{x\in\{0,1\}^n} \sum_{k=0}^m \binom{m}{k} (-2^{-n})^{m-k} \braa{x^{\ot(k+1)}} \mc{C}^{\ot(k+1)} \kett{\rho^{\ot(k+1)}} \\
&= (2^n+1)^m \sum_{k=0}^m \sum_{T\in\Sigma_{k+1,k+1}} \binom{m}{k} (-2^{-n})^{m-k} \frac{1}{\prod_{\ell=0}^{k-1} (2^n + 2^\ell)} \braakett{R_T}{\rho^{\ot(k+1)}}\notag \\
&= (2^n+1)^m \sum_{k=0}^m \binom{m}{k} (-2^{-n})^{m-k} \prod_{\ell=0}^{k-1} \frac{2^\ell+1}{2^\ell + 2^n},
\end{align*}
where the third line follows from \cref{eq:state_av_cliff}.
In the final step we used that $\abs{\Sigma_{t,t}} = \prod_{\ell=0}^{t-2} (2^\ell+1)$ and $\braakett{R_T}{S^{\ot t}}=0$ for every stabilizer state~$\ket{S}$ and any~$T \in \Sigma_{t,t}$, by~\cite[Theorem~4.9 and Eq.~(4.10)]{gross2021schur}.
\end{proof}

\begin{theorem}[Reprint of \cref{thm:clifford_tails}]
Consider shadow estimation with the $n$-qubit Clifford group as circuit set, any $n$-qubit stabilizer state~$\rho=\dens{S}$, and the observable~$O = \dens{S} - 2^{-n} I$, so that $\tr(O^2) \leq 1$.
For every~$n$, let~$\mathbf{X}_n=\mathbf{X}$ denote the associated random variable.
Suppose that the sequence $\mathbf{X}_n$ satisfies a tail bound of the form
\begin{align}\label{eq:tail hypothesis restated}
  \Pr\mleft( \abs{ \mathbf{X}_n - \mathbb{E}(\mathbf{X}_n) } \geq t \mright) \leq A \exp\mleft( - \frac {t^\beta} {B_n} \mright),
\end{align}
for constants $A,\beta>0$ and a positive sequence $(B_n)$.
Then we have that $B_n = \tilde\Omega(2^{\beta n/4})$.
\end{theorem}

\begin{proof}
We start by lower bounding the $n$-th moments of~$\mathbf{X}_n$, using the formula in \cref{eq:m_th moment appendix}:
\begin{align}
\nonumber
  \mathbb{E}(\mathbf{X}_n^n)
&= (2^n+1)^n \sum_{k=0}^n \binom{n}{k} (-1)^{n-k} 2^{-n(n-k)} \prod_{\ell=0}^{k-1} \frac{2^\ell+1}{2^\ell + 2^n} \\
\nonumber
&= \bracks*{ 1 - n 2^{-n} \frac{2^{n-1} + 2^n}{2^{n-1}+1} + \sum_{k=0}^{n-2} \binom{n}{k} (-1)^{n-k} 2^{-n(n-k)} \prod_{\ell=k}^{n-1} \frac{2^\ell + 2^n}{2^\ell+1} } \prod_{\ell=0}^{n-1} \frac{(2^n+1)(2^\ell+1)}{2^n + 2^\ell} \\
\nonumber
&\geq \bracks*{ 1 - \frac{3n}{2^n+2} - \sum_{k=0}^{n-2} \binom{n}{k} 2^{-n(n-k)} \prod_{\ell=k}^{n-1} \frac{2^\ell + 2^n}{2^\ell+1} } \prod_{\ell=0}^{n-1} \frac{(2^n+1)(2^\ell+1)}{2^n + 2^\ell} \\
\nonumber
&= \bracks*{ 1 - \frac{3n}{2^n+2} - \sum_{k=0}^{n-2} \binom{n}{k} \prod_{\ell=k}^{n-1} \frac{2^\ell + 2^n}{2^{\ell+n}+2^n} } \prod_{\ell=0}^{n-1} \frac{(2^n+1)(2^\ell+1)}{2^n + 2^\ell} \\
\nonumber
&\geq \bracks*{ 1 - \frac{3n}{2^n+2} - 2^n \frac{(2^{n-2} + 2^n)(2^{n-1} + 2^n)}{(2^{2n-2}+2^n)(2^{2n-1}+2^n)} } \prod_{\ell=0}^{n-1} \frac{(2^n+1)(2^\ell+1)}{2^n + 2^\ell} \\
\nonumber
&\geq \bracks*{ 1 - \frac{3n+15}{2^n} } \prod_{\ell=0}^{n-1} \frac{(2^n+1)(2^\ell+1)}{2^n + 2^\ell} \\
\label{eq:Xnn lower bound}
&\geq \frac12 \prod_{\ell=0}^{n-1} \frac{(2^n+1)(2^\ell+1)}{2^n + 2^\ell}
\geq 2^{\frac {n(n-3)}2}
\geq 2^{n^2/4},
\end{align}
where the second inequality holds because $\prod_{\ell=k}^{n-1} \frac{2^\ell + 2^n}{2^{\ell+n}+2^n}$ is monotonically increasing with~$k$ and the last two inequalities are valid for~$n\geq7$.

We will now show that this super-exponential growth is not commensurate with a tail bound of the form we presumed to exist, by standard arguments.
We first note that \cref{eq:tail hypothesis} imposes strong conditions on the growth of the absolute centered moments restated.
By using that for any nonnegative random variable $\mathbf{Z}$ we have $\mathbb{E}(\mathbf{Z}) = \int_0^\infty \mathbb{P}(\mathbf{Z}\!\geq\! s)\,ds$, we can write the absolute centered $n$-th moment of $\mathbf{X}_n$ as
\begin{align}
\nonumber
  \mathbb{E} \parens*{ \abs{\mathbf{X}_n - \mathbb{E}(\mathbf{X}_n)}^n }
&= \int_0^\infty ds \, \Pr\mleft( \abs{\mathbf{X}_n - \mathbb{E}(\mathbf{X}_n)}^n \geq s \mright) \\
\nonumber
&= \int_0^\infty dt \, n \, t^{n-1} \Pr\mleft( \abs{\mathbf{X}_n - \mathbb{E}(\mathbf{X}_n)} \geq t \mright) \\
\nonumber
&\leq \int_0^\infty dt \, n \, t^{n-1} A \exp\mleft( - \frac {t^\beta} {B_n} \mright) \\
&= \frac {A n B_n^{n/\beta}} \beta \Gamma\mleft( \frac n \beta \mright),
\label{eq:Xnn centered gamma bound}
\end{align}
using the relationship between the moments of the ``stretched exponential'' function and the $\Gamma$-function.
We can relate the above to the ordinary moments by the Minkowski inequality,
\begin{align*}
  \parens*{ \mathbb{E} \abs{\mathbf{X}_n}^n }^{1/n}
\leq \parens*{ \mathbb{E} \abs{\mathbf{X}_n - \mathbb{E}(\mathbf{X}_n)}^n }^{1/n} + \mathbb{E}(\mathbf{X}_n)
\leq \parens*{ \mathbb{E} \abs{\mathbf{X}_n - \mathbb{E}(\mathbf{X}_n)}^n }^{1/n} + 1,
\end{align*}
where in the last step we used that $\mathbb{E}(\mathbf{X}_n) = \tr(O \rho) = 1 - 2^{-n}$.
From this and \cref{eq:Xnn lower bound,eq:Xnn centered gamma bound}, we obtain the inequality
\begin{align*}
  2^{n/4}
\leq \parens*{ \mathbb{E} \abs{\mathbf{X}_n - \mathbb{E}(\mathbf{X}_n)}^n }^{1/n} + 1
\leq B_n^{1/\beta} \bracks*{ \frac {A n} \beta \Gamma\mleft( \frac n \beta \mright) }^{1/n} + 1.
\end{align*}
Using Stirling's approximation for the $\Gamma$ function it follows that we must have $B_n^{1/\beta} = \tilde\Omega(2^{n/4})$, that is, $B_n = \tilde\Omega(2^{\beta n/4})$.
\end{proof}

\begin{corollary}\label{cor:clifford_as_tails}
Consider shadow estimation with the $n$-qubit Clifford group as circuit set, any $n$-qubit stabilizer state~$\rho=\dens{S}$, and the observable~$O = \dens{S} - 2^{-n} I$.
Let~$\mathbf{X}_n=\mathbf{X}$ denote the associated random variable.
Then, for all~$m\geq0$,
\begin{equation*}
  \lim_{n\to\infty} \mathbb{E}(\mathbf{X}_n^m)
= \sum_{k=0}^m \binom{m}{k} (-1)^{m-k} \prod_{\ell=0}^{k-1}(2^\ell+1).
\end{equation*}
For~$m\geq 6$, the right-hand side this can be lower bounded as follows:
\begin{align*}
  \lim_{n\to\infty} \mathbb{E}(\mathbf{X}_n^m)
\geq 2^{\frac{m(m-1)}{2}}
= \Omega(2^{m^2/2}).
\end{align*}
\end{corollary}
\begin{proof}
It follows from \cref{eq:m_th moment appendix} that for any fixed~$m$ we have that
\begin{align*}
  \lim_{n\to\infty} \mathbb{E}(\mathbf{X}_n^m)
= \sum_{k=0}^m \binom{m}{k} (-1)^{m-k} \parens*{ \lim_{n\to\infty} \frac {(2^n+1)^m}{2^{n(m-k)} \prod_{\ell=0}^{k-1} (2^\ell + 2^n)} } \prod_{\ell=0}^{k-1}(2^\ell+1)
= \sum_{k=0}^m \binom{m}{k} (-1)^{m-k} \prod_{\ell=0}^{k-1}(2^\ell+1),
\end{align*}
which confirms the formula for the limiting moments.
We now prove the lower bound by similar estimates as in the proof of Theorem $6$:
\begin{align*}
  \lim_{n\to\infty} \mathbb{E}(\mathbf{X}_n^m)
&= \bracks*{ 1 - \frac m {2^{m-1}+1} + \sum_{k=0}^{m-2} \binom{m}{k} (-1)^{m-k} \frac 1 {\prod_{\ell=k}^{m-1}(2^\ell+1)} } \prod_{\ell=0}^{m-1}(2^\ell+1) \\
&\geq \bracks*{ 1 - \frac m {2^{m-1}+1} - \sum_{k=0}^{m-2} \binom{m}{k} \frac 1 {\prod_{\ell=k}^{m-1}(2^\ell+1)} } \prod_{\ell=0}^{m-1}(2^\ell+1) \\
&\geq \bracks*{ 1 - \frac m {2^{m-1}+1} - \frac {2^m} {(2^{m-1}+1)(2^{m-2}+1)} } \prod_{\ell=0}^{m-1}(2^\ell+1) \\
&\geq \bracks*{ 1 - \frac{2m+8} {2^m} } \prod_{\ell=0}^{m-1}(2^\ell+1) \\
&\geq \frac12 \prod_{\ell=0}^{m-1}(2^\ell+1),
\end{align*}
where we used that $(2m+8)/2^m \leq 1/2$ for $m\geq6$.
The desired lower bound follows from the above and the estimate
\begin{equation*}
  \frac12 \prod_{\ell=0}^{m-1}(2^\ell+1)
= \prod_{\ell=1}^{m-1}(2^\ell+1)
\geq \prod_{\ell=1}^{m-1} 2^\ell
= 2^{\sum_{\ell=1}^{m-1} \ell}
= 2^{\frac{m(m-1)}2}.
\qedhere
\end{equation*}
\end{proof}

\section{Technical lemmas}
Finally we prove some useful, but less interesting, technical statements that were used in the proofs of Theorems $3$ and $4$.
Recall that the commutant of the fourth tensor power action of the Clifford group is parameterized by the set $\Sigma_{4,4} = S_4 \cup \hat{S}_3$.

\begin{lemma}\label{lem:v_bound 1}
For every state~$\rho$, traceless observable~$O$, and~$T\in \Sigma_{4,4}$, we have
\begin{align*}
  \abs{\braakett{R_T}{(O\otimes \rho)\tn{2}}} \leq \tr(O^2).
\end{align*}
\end{lemma}
\begin{proof}
We first prove the claim for $\pi\in S_4\subset \Sigma_{4,4}$.
Clearly,
\begin{align*}
  \braakett{R_\pi}{(O\otimes \rho)\tn{2}}
= \tr \mleft( R_\pi (O \otimes \rho \otimes O \otimes \rho) \mright)
\end{align*}
is an arbitrary product of traces of products of the operators~$O$ and~$\rho$, subject only to the constraint that each of $O$ and $\rho$ should appear exactly twice.
Using $\tr(\rho)=1$, $\tr(O)=0$, and the cyclicity of the trace, the following estimates imply the desired bound for any $\pi\in S_4$,
\begin{align*}
  \abs{\tr(O \rho)}^2 &\leq \tr(O^2) \tr(\rho^2) \leq \tr(O^2), \\
  \abs{\tr(O^2) \tr(\rho^2)} &\leq \tr(O^2), \\
  \abs{\tr(O^2 \rho)} &\leq \tr(O^2), \\
  \abs{\tr(O^2 \rho^2)} &\leq \tr(O^2), \\
  \abs{\tr(O \rho O \rho)} &\leq \tr(O^2 \rho) \leq \tr(O^2),
\end{align*}
which follow from the Cauchy-Schwarz inequality and~$0\leq \rho\leq I$.

It remains to prove the claim for $T \in \hat{S}_3$.
Recall that~$T\in \hat{S}_3~$ implies~$R_T = R_\pi \Pi_4$ for some~$\pi \in S_3 \subseteq S_4$, with~$\Pi_4 = 2^{-n} \sum_{P \in \mathbb{P}_n} P^{\ot 4}$, hence
\begin{align*}
  \braakett{R_T}{(O\otimes \rho)\tn{2}}
= \tr \mleft( \Pi_4 R_{\pi^{-1}} (O \otimes \rho \otimes O \otimes \rho) \mright).
\end{align*}
Clearly, $\Pi_4$ commutes with arbitrary permutations, while $(O \otimes \rho \otimes O \otimes \rho)$ in particular commutes with the permutation~$(1 3)$.
This means that it suffices to verify the claim for $\pi \in \{ e, (1 2), (1 3), (1 2 3) \}$.
For $\pi=e$, the identity permutation, noting that~$\abs{\Tr(\rho P)} \leq 1$ gives
\begin{align*}
  \abs*{ \tr \mleft(\Pi_4 (O \otimes \rho \otimes O \otimes \rho)\mright)}
= 2^{-n} \sum_{P \in \mathbb{P}_n} \parens*{ \tr (P O) }^2 \big(\tr (P \rho)\big)^2
\leq 2^{-n} \sum_{P \in \mathbb{P}_n} \parens*{ \tr (P O) }^2
= \tr(O^2),
\end{align*}
where the last follows by Parseval's identity since $\{2^{-n/2} P\}_{P \in \mathbb{P}_n}$ is an orthonormal basis of the space of $n$-qubit operators.
For~$\pi=(1 2)$, by the Cauchy-Schwarz inequality and again Parseval's identity,
\begin{align*}
  \abs*{ \tr \mleft( \Pi_4 R_{(1 2)} (O \otimes \rho \otimes O \otimes \rho) \mright)}
&= 2^{-n} \abs*{ \sum_{P \in \mathbb{P}_n} \tr(OP\rho P) \tr(OP) \tr(\rho P) }\\
&\leq \norm O_\infty 2^{-n} \sum_{P \in \mathbb{P}_n} \abs*{\tr(OP)} \abs*{\tr(\rho P)} \\
&\leq \norm O_\infty \norm O_{\HS} \norm \rho_{\HS}
\leq \norm O_{\HS}^2 = \tr (O^2).
\end{align*}
For $\pi=(1 3)$,
\begin{align*}
  \abs*{ \tr \mleft( \Pi_4 R_{(1 3)} (O \otimes \rho \otimes O \otimes \rho) \mright)}
= 2^{-n} \abs*{ \sum_{P \in \mathbb{P}_n} \tr (OPOP) \tr(\rho P) \tr(\rho P) }
\leq \norm{O}_{\HS}^2 \norm{\rho}_{\HS}^2
\leq \tr(O^2),
\end{align*}
where we used Parseval's identity and that $\abs*{\tr (OPOP)} \leq \norm{O}_{\HS} \norm{P O P}_\HS = \norm{O}_{\HS}^2$ by the Cauchy-Schwarz inequality.
Finally, for~$\pi=(1 3 2)$,
\begin{align*}
  \abs*{ \tr \mleft( \Pi_4 R_{(1 2 3)} (O \otimes \rho \otimes O \otimes \rho) \mright)}
&= 2^{-n} \abs*{ \sum_{P \in \mathbb{P}_n} \tr (OP\rho POP) \tr(\rho P) }
\leq 2^{-n} \sum_{P \in \mathbb{P}_n} \tr (OP\rho PO),
= \tr(O^2),
\end{align*}
where the first inequality holds because $\abs{\tr(AB)} \leq \tr(A)\norm{B}_\infty$ when $A$ is positive semidefinite and $B$ is Hermitian (this is H\"older's inequality), and the final step is due to~$2^{-n} \sum_{P \in \mathbb{P}_n} P \rho P = \Tr(\rho) I = I$ (Schur's lemma).
This concludes the proof of the lemma.
\end{proof}

\begin{lemma}\label{lem:v_bound 2}
For every $x,\hat x\in\{0,1\}^n$, we have
\begin{align*}
  \braakett{x\tn{2}\otimes \hat{x}\tn{2}}{R_T}
= \begin{cases}
  1 & \text{ if } T \in \{e, (12), (34), (12)(34), T_4, (1 2) T_4 \}, \\
  \delta_{x,\hat x} & \text{ otherwise }.
\end{cases}
\end{align*}
where $T = \pi T_4$ for $\pi\in S_3$ denotes the subspace corresponding to $R_T = R_\pi \Pi_4$, see the discussion surrounding \cref{eq:Pi_4}.
\end{lemma}
\begin{proof}
For $\pi\in S_4$ the claim is clear, since $\braakett{x\tn{2}\otimes \hat{x}\tn{2}}{R_\pi}$ contains a factor~$\delta_{x,\hat x}$ unless~$\pi \in \{e, (12), (34), (12)(34)\}$.
Thus it remains to consider $T \in \hat{S}_3$ and hence $R_T = R_\pi \Pi_4$ for~$\Pi_4 = 2^{-n} \sum_{P \in \mathbb{P}_n} P^{\ot 4}$ and some~$\pi \in S_3 \subseteq S_4$.
Then,
\begin{align*}
  \braakett{x\tn{2}\otimes \hat{x}\tn{2}}{R_T}
&= 2^{-n} \sum_{P \in \mathbb{P}_n} \tr \bracks*{ \parens*{ \dens{x}^{\ot 2} \ot \dens{\hat{x}}^{\ot 2} } R_\pi P^{\ot 4} } \\
&= 2^{-n} \sum_{P \in \mathbb{P}_n} \tr \bracks*{ \parens*{ \dens{x}^{\ot 2} \ot \dens{\hat{x}} } R_\pi P^{\ot 3} } \bra{\hat{x}} P \ket{\hat{x}} \\
&= 2^{-n} \sum_{b \in \{0,1\}^n} (-1)^{\hat{x} \cdot b} \tr \bracks*{ \parens*{ \dens{x}^{\ot 2} \ot \dens{\hat{x}} } R_\pi (Z^b)^{\ot 3} } \\
&= \tr \bracks*{ \parens*{ \dens{x}^{\ot 2} \ot \dens{\hat{x}} } R_\pi },
\end{align*}
where from the second line onward by a slight abuse of notation we think of $R_\pi$ as an operator on $((\mathbb C^2)^{\ot n})^{\ot 3}$.
It is clear that the above expression is zero for $x \neq \hat{x}$ unless $\pi \in \{e,(12)\}$.
\end{proof}

\begin{lemma}\label{lem:t_gate}
Let $\mc{T}$ denote the quantum channel acting by the $\T$-gate $\T = \begin{psmallmatrix} 1& 0 \\ 0 & e^{i\pi/4}\end{psmallmatrix}$ on the first qubit of an $n$-qubit state.
Then we have, for every $T,T'\in\hat{S}_3$, that
\begin{align*}
  \braa{R_T} \mc{T}^{\ot4} \kett{R_{T'}}
= \parens*{ \braakett{r_T}{r_{T'}} - 4 } \braakett{r_T}{r_{T'}}^{n-1}
= \begin{cases}
  (2^4 - 4) 2^{4(n-1)} = \frac34 2^{4n} & \text{ if  }\; T=T', \\
  \leq (2^3-4) 2^{3(n-1)} = \frac12 2^{3n} & \text{ if  }\; T\neq T'.
\end{cases}
\end{align*}
\end{lemma}
\begin{proof}
The second identity follows from the first, since $\braakett{r_T}{r_{T'}} \in \{2^2, 2^3, 2^4\}$, with $\braakett{r_T}{r_{T'}}=2^4$ if and only if $T=T'$.
For the first identity we only need to show that
\begin{align*}
  \braa{r_T} \mc{T}^{\ot4} \kett{r_{T'}} = \braakett{r_T}{r_{T'}} - 4.
\end{align*}
Let us write~$T = \pi T_4$ and~$T' = \pi' T_4$ for $\pi,\pi'\in S_3$.
Then,
\begin{align*}
  \braa{r_T} \mc{T}^{\ot4} \kett{r_{T'}}
= \frac14 \sum_{P,P'\in\mathbb{P}_1} \tr \mleft( P^{\ot 4} r_{\pi^{-1}} \T^{\ot4} r_{\pi'} (P')^{\ot4} (\T^{\dagger})^{\ot4} \mright)
= \frac14 \sum_{P,P'\in\mathbb{P}_1} \tr \mleft( \parens*{ P \T P' \T^\dagger }^{\ot4} r_{\pi' \pi^{-1}} \mright).
\end{align*}
Noting that
\begin{align*}
  \T I \T^\dagger = I,
\qquad
  \T X \T^\dagger = \frac{X + Y}{\sqrt2},
\qquad
  \T Y \T^\dagger = \frac{-X + Y}{\sqrt2},
\qquad
  \T Z \T^\dagger = Z,
\end{align*}
we find that
\begin{align*}
\sum_{P'\in\mathbb{P}_1} \parens*{ \T P' \T^\dagger }^{\ot4}
&= I^{\ot4} + \parens*{ \frac{X + Y}{\sqrt2} }^{\ot4} + \parens*{ \frac{-X + Y}{\sqrt2} }^{\ot4} + Z^{\ot4} \\
&= \sum_{P'\in\mathbb{P}_1} \parens*{ P' }^{\ot4} + \frac12 \bigl( -X^{\ot4} - Y^{\ot4} + X\ot X\ot Y\ot Y + X\ot Y\ot X\ot Y + X\ot Y\ot Y\ot X \\
&\hspace{14em}\;\;\, + Y\ot X\ot X\ot Y + Y\ot X\ot Y\ot X + Y\ot Y\ot X\ot X \bigr).
\end{align*}
Noting that $r_{\pi' \pi^{-1}}$ acts as the identity on the fourth qubit and using the product rules of the Pauli group, we obtain
\begin{align*}
  \braa{r_T} \mc{T}^{\ot4} \kett{r_{T'}}
&= \braakett{r_T}{r_{T'}}
- \frac14  \tr \mleft( I^{\ot4} r_{\pi' \pi^{-1}} \mright)
- \frac14  \tr \mleft( \parens*{ I \ot Z \ot Z \ot I} r_{\pi' \pi^{-1}} \mright) \\
&\qquad\qquad\quad\; - \frac14  \tr \mleft( \parens*{ Z \ot I \ot Z \ot I} r_{\pi' \pi^{-1}} \mright)
- \frac14  \tr \mleft( \parens*{ Z \ot Z \ot I \ot I} r_{\pi' \pi^{-1}} \mright) \\
&= \braakett{r_T}{r_{T'}}
- \tr \mleft( ( \dens{0}^{\ot3} \ot I) r_{\pi' \pi^{-1}} \mright)
- \tr \mleft( ( \dens{1}^{\ot3} \ot I) r_{\pi' \pi^{-1}} \mright) \\
&= \braakett{r_T}{r_{T'}} - 4.
\end{align*}
where the second step follows from $\dens0 = (I+Z)/2$ and $\dens1 = (I-Z)/2$.
\end{proof}

\end{widetext}
\end{document}